\documentclass{article}
\usepackage{graphicx} 
\usepackage[margin=2.5cm]{geometry}
\usepackage{amsmath,amsfonts,amssymb,amsthm}
\usepackage{hyperref}
\usepackage{color}

\newtheorem{lemma}{Lemma}
\newtheorem{theorem}{Theorem}
\newtheorem{example}{Example}
\newtheorem{remark}{Remark}

\usepackage{stmaryrd}
\usepackage{amsmath}
\usepackage{amsthm}
\usepackage{amsfonts}
\usepackage{hyperref}

\newtheorem{proposition}{Proposition}
\newtheorem{problem}{Problem}

\usepackage{algorithm}
\usepackage{adjustbox}
\usepackage[noend]{algpseudocode}
\usepackage{booktabs}
\usepackage{subcaption}
\usepackage{tikz}
\usetikzlibrary{quantikz2}
\usepackage{xspace}

\usepackage{cases}
\usepackage{color}

\usepackage{natbib}
 \bibpunct[, ]{(}{)}{,}{a}{}{,}%

\newcommand{\CNOT}{{\sc cnot}}
\newcommand{\CNOTs}{{\sc cnot}s}
\newcommand{\RZ}{{rotation-Z}}
\newcommand{\NP}{{\sf NP}}
\newcommand{\Te}{\mathbb{T}}

\usepackage{todonotes}

\newcommand{\pbname}{Poly-Z min-depth circuit synthesis problem\xspace}
\newcommand{\pbac}{PZDCS\xspace}
\newcommand{\arpbac}{PZDCS-Top\xspace}

\newcommand{\cnot}[2]{\mathtt{cx}_{#1#2}}
\newcommand{\rz}[1]{\mathtt{rz}_{#1}}

\newcommand{\VT}{\texttt{Tempura}} 

\newcommand{\TB}{\texttt{HUBO\_BLP\_TB\_H}}
\newcommand{\TBFM}{\texttt{HUBO\_BLP\_TBFM\_H}}
\newcommand{\CDP}{\texttt{HUBO\_BLP}}
\newcommand{\Halmos}{\qed}

\begin{document}
\author{Dylan Laplace Mermoud $^{a}$, Dimitri Watel $^{b,c}$, Sourour Elloumi $^{b,c,d}$, Andrea Simonetto $^{c,d}$ \\
\footnotesize
$^a$ Sorbonne Université, CNRS, LIP6, F-75005 Paris, France \\
\footnotesize $^b$ ensIIE, 1 place de la Résistance, Évry-Courcouronnes, France\\
\footnotesize $^c$ CEDRIC lab, CNAM, Paris, France\\
\footnotesize $^d$ Unité des Mathématiques Appliquées, ENSTA, Institut Polytechnique de Paris,  Palaiseau, France \\
}

\title{Circuit Design Problems in Binary Polynomial Quantum Optimization}

\date{}

\maketitle

\begin{abstract}
Solving combinatorial optimization problems on near-term quantum computers mainly relies on their quadratic unconstrained binary formulation. Lately, binary polynomial formulations start to take the center stage, having a better scaling with respect to the number of used qubits. However, while quadratic forms are straightforward to implement on quantum circuits, polynomial forms incur in circuit depth overheads, due to the lack of $n$-qubit gates, for $n\geq 3$, on current hardware. In this paper, we study the problem of designing the shallowest quantum circuit for a given polynomial form. We formulate it as a classical binary linear problem and we prove the \textsf{NP}-completeness of a variant of it. On the algorithmic side, we offer a scalable and performant heuristic method to tackle it, \VT, as well as faster binary linear problem formulations, which offer computational feasibility while sacrificing optimality.  Our performance evaluation shows that the algorithms yield far shorter circuits than the state-of-the-art approaches, and \VT\ is mostly comparable with them in terms of runtime.    
\end{abstract}


\section{Introduction}\label{sec:Intro}

The current state of the art in solving combinatorial optimization problems on near-term quantum computers favors their reformulation as quadratic unconstrained binary optimization (QUBO) problems. This is mainly due to a popular quantum variational algorithm, QAOA~\citep{QAOA,Zhou2018}, which was initially thought for these problems, and for the ease in which quadratic forms map into a quantum circuit. Such quantum circuits can then be encoded into current hardware. The use of QUBO formulations is not very restrictive from a mathematical standpoint. We know that {\sc{maxcut}} is \textsf{NP}-Hard and a QUBO. As such, we could polynomially reduce many interesting problems (i.e., all known \textsf{NP}-Hard problems) to it~\citep{arora2009computational}. However, this way of doing hides critical complexities in how many additional variables we may need to consider, which maps directly into the number of used qubits. Quadratization has also been considered, with similar overheads~\citep{mandal2020compressed,GOQAOA}. Qubits, the quantum equivalent to classical logical bits, as well as the number of operations we perform on them, i.e., the number of gates, represent the current bottleneck in implementing quantum optimization algorithms at scale. As such, compact formulations of optimization problems are to be preferred to loose ones, with many variables and redundancies.

In an effort to reduce the number of variables and qubits, binary polynomial problems have started to take a more central role. We often call these problems HUBOs, for higher-order unconstrained binary optimization. Several works have demonstrated how their formulation is more compact in terms of number of variables~\citep{HOBO}, and sometimes it even needs less gates when applied to multi-layer QAOA~\citep{koch2026resource}. See also interesting recent applications in~\citep{uotila2025higher,romero2025bias}. In addition, while QUBOs have been solved classically up to thousands of variables~\citep{Rehfeldt2023} and a quantum advantage seems unlikely in the medium term, HUBOs are expected to offer a quicker path to quantum utility~\citep{Shaydulin2024}. 

The key disadvantage of using HUBO formulations is their quantum circuit depth. If QAOA can be used also for HUBOs, when applied to a polynomial of maximum degree $d$, it generates quantum gates featuring $d$-qubits, which cannot be easily encoded in current quantum hardware if $d \geq 3$. Typically, one would need to decompose such gates into one and two-qubit gates, yielding a circuit depth overhead~\citep{shende2004lowerbound}. As such, QUBOs and HUBOs typically stand at the opposite side of the trade-off: short circuits for many qubits, or long circuits for few qubits\footnote{As said, \cite{koch2026resource} claim less depth for HUBO problems than QUBO for specific formulations. This can be the case if QAOA requires less layers for the HUBO than the QUBO to obtain a good approximate solution. However, very often each layer's depth is longer in the HUBO case than in the QUBO case.}.

Said so, HUBOs are still very interesting from their quantum utility potential. In this paper, we start from the HUBO polynomial form (with $n$ variables, max $d$ degree, and $m$ monomials) and we study how to use classical combinatorial optimization to map the resulting $d$-qubit gates into parametrized \RZ\ and \CNOT\ gates. Specifically, we focus on how to minimize the depth of the resulting quantum circuit at most as possible. We start from our preliminary work~\citep{verchere2023optimizing} which formulated the HUBO circuit design problem as a binary linear optimization problem and proposed an heuristic method to tackle it. Here, we offer the following main contributions: \emph{(i)} we shed light on the complexity of this problem, by proving a {\textsf{NP}}-completeness result of a variant of it; \emph{(ii)} we formulate several binary linear problems that encode the problem with a various degree of optimality vs computational feasibility, which are computationally better than our preliminary work; \emph{(iii)} we propose an efficient, scalable, and performant heuristic algorithm \VT, which allows us to compile large circuits featuring thousands of monomials and has $O(m n)$ worst-case, but in practice better than $O(m d)$, depth, going beyond the naive $O(m 2^d)$ obtained in our preliminary work; \emph{(iv)} we demonstrate practical performance with respect to two common state-of-the-art solvers (Qiskit and \texttt{tket}), showing far shorter circuits and similar runtimes (for \VT).

\subsection{Related work}

The general problem of transferring a quantum algorithm into a quantum circuit goes under the name of compilation. In the most wholistic sense, we can think of a quantum algorithm featuring $n$-qubits as a unitary matrix $U\in\mathbb{C}^{2^n \times 2^n}$ mapping input to output. Compilation is then the task of decomposing and transferring such matrix into operations (gates) that we can implement on a specific hardware. Typical constraints that we have are the gate set, how the qubits can be arranged (topology or connectivity), and the total number of operations (circuit length) limitations. 

Since the celebrated Solovay-Kitaev theorem~\citep{nielsenbook}, quantum compiling has been a rich research area. Works have investigated how to efficiently map different gates into canonical (universal) gate sets up to an arbitrary accuracy, or how to ``place'' the target circuit onto the real connectivity-limited hardware. See for instance~\citep{madden2022best,nannicini2022optimal} and the many references therein. More recently, machine learning techniques have also been used for this task~\citep{moro2021quantum,yan2024quantum,theissinger2026beyond}. In fact, quantum compilation off-the-shelf softwares, such as \texttt{tket}~\citep{tket}, \texttt{bqskit}~\citep{doecode58510}, and the full software stack Qiskit~\citep{Qiskit2019} encode most of the current research in proving a scalable, fast, general-purposes compilation procedure, often based on fast routines and sacrificing circuit depth.

The goal of compiling a circuit given a gate set and connectivity constraints can be thought in the lens of operations research. Mixed-integer linear formulations have then appeared to tackle specific instances of the compilation and mapping problem, e.g.,~\citep{nagarajan2021quantumcircuitopt,nannicini2022optimal,verchere2023optimizing,nagarajan2025provably,jung2025new}, all featuring optimal solutions but unmanageable classical overheads after $n>10$ qubits. In particular, \cite{nannicini2022optimal} present an integer linear program how to optimally place and route the gates on a quantum computer via limited connectivity, given a certain full-connectivity quantum circuit. The works of~\citep{nagarajan2021quantumcircuitopt,nagarajan2025provably,jung2025new}, they focus instead on decomposing and compiling a given unitary matrix in a given gate set with the goal of minimizing depth. 

In this paper, we focus on a specific instance of the compilation problem: how to map diagonal multi-qubit $Z$ gates (which are the $d$-qubit gates stemming from the HUBO polynomial form) into parametrized \RZ\ and \CNOT\ gates. Universal baselines for the length of the circuits are the $O(2^n)$ lower bound in exact compilation and the $O(1/\epsilon)$ in a $\epsilon$-approximate one~\citep{zylberman2025efficient}. However, here we do not have a general unitary and we can compile shorter circuits if the degree $d$ of the polynomial is less than $n$, as we will see. 

Given the combinatorial nature of the compilation task, it seems natural to ask whether some instances of the compilation problem are \NP-Complete. In general, the \NP-completeness of quantum compilation is an open problem. \cite{botea2018complexity} provides probably one of the first proofs in this direction for a particular instance. Close in spirit to our task is the work of \cite{Amy2019} and the very recent \cite{acuaviva2026cnotdistance}. In this paper, we provide a \NP-Complete result for a topologically-constrained version of compilation, close to our formulation.

\subsection{Organization}

The next section introduces the problem formulation, some useful notation, and definitions. In Section~\ref{sec:np-completeness}, we prove the \NP-completeness of an instance of our compilation problem under some conditions about the connectivity of the qubits. Section~\ref{sec:TPA} describes our efficient template and pruning algorithm to the compilation task, while Section~\ref{sec:ILP} presents the optimal binary linear formulation. Then, in Section~\ref{sec:results}, we discuss performance evaluations of the different approaches, and we finish the paper with some concluding remarks. 

\section{Problem description}

We are interested in binary combinatorial optimization problems of the form
\begin{equation}
\min_{z \in \{-1, 1\}^n} P(z), 
\end{equation}
where $P(z)$ is a polynomial in the $\{-1,1\}$ variables $z$. To make the dependence explicit, we introduce the binary selection variables $\mu_j \in \{0,1\}$ for variable $z_j$, we stack them into the selection vector $\mu \in \{0,1\}^n$ and consider the set of all possible $\mu$'s, $M =\{\mu | \mu \in \{0,1\}^n \}$ of cardinality $2^n$. In this way, we can write the polynomial $P(z)$ as,
\begin{equation}
P(z) =  \sum_{\mu \in M} \alpha_{\mu} \prod_{j = 1}^n z_j^{\mu_j},
\end{equation}
where \(\alpha_\mu\) is a real number for every \(\mu \in M\). We denote with $d$ the maximum degree among all the monomials, $d = \max_{\mu \in M} \{ {\mathbf{1}}_{\alpha_{\mu}}\sum_j \mu_j\}$, and by $m$ the number of monomials, i.e., $m = \sum_{\mu \in M} {\mathbf{1}}_{\alpha_{\mu}}$, where ${\mathbf{1}}_{\alpha_{\mu}}$ is $1$ if and only if $\alpha_{\mu} \neq 0$. Typically $m \ll 2^n$. If $d = 2$, then the polynomial is quadratic. 

To find an approximate \(z\) minimizing \(P\) over \(\{-1, 1\}^n\), one can resort to the quantum variational heuristic \textsc{QAOA}. To do so, we associate to each variable $z_i$ a qubit, we introduce the Pauli-Z gate, or $Z$ operator, as,
$
Z = \begin{bmatrix} 1 & 0\\ 0 & -1\end{bmatrix},
$
and the problem Hamiltonian,
\[
H_P(Z) = \sum_{\mu \in M} \alpha_\mu \bigotimes_{j = 1}^n Z_j^{\mu_j}. 
\]
Here, we have defined \(Z_j\) the operator \(Z\) acting on the \(j\)-th qubit. We have now to build the quantum circuit implementing the exponential of the matrix \(U_P(\gamma) = \exp{(-\jmath \gamma H_P(Z))}\), where $\gamma$ is a free rotation parameter that will be optimized over classically, and $\jmath = \sqrt{-1}$. 
In the following, we let $\gamma = 1$ without loss of generality, we define $U_P = U_P(1)$ and load into the $\alpha$ coefficient the angle $\gamma$ information. 
The matrix \(\exp{(-\jmath H_P(Z))}\) is a diagonal matrix of dimension $2^n \times 2^n$. Up to a global phase, we can implement this unitary with \CNOT\ and \RZ\ gates~\citep{nielsenbook}, which are the typical ones available in hardware. We will indicate the 2-qubit \CNOT\ gate as \(\cnot{i}{j}\) if it is controlled by qubit \(i\), targeting the qubit \(j\). We adopt the notation that a \CNOT\ is downward if $j>i$, and upward otherwise. 
The 1-qubit \RZ\ gates, will be indicated as $\rz{}$ (leaving the angle dependency hidden). Whenever the dependencies of the qubit is important, we will say $\rz{i}$. Often we will also say $\rz{i,j,k, \dots}$ to stress that the \RZ\ gate is implementing the coefficient $\alpha_{i,j,k, \ldots}$. In the quantum circuit, we will simply indicate $\alpha_{i,j,k, \ldots}$. This notation is useful, since it allows us to present parametric circuits (that are valid for all the choices of non-zero coefficients), rather than fixing them to specified values. 

Our aim is to design such circuits while minimizing their depth. We define here depth as the count of time steps needed to execute all the gates in the quantum circuit, when one time step is the time it takes to execute any gate. If $N$ gates can be executed in parallel, then only one time step is needed. Formally, we introduce the following problem. 

\begin{problem}[\pbname (\pbac)]
Given a Pauli-Z polynomial
\[
H_P(Z) = \sum_{\mu \in M} \alpha_\mu \bigotimes_{j=1}^n Z_j^{\mu_j}, 
\]
design a quantum circuit implementing \(U_P = \exp{(-\jmath H_P(Z))}\) of minimum depth, using only \CNOT\ and \RZ\ gates. \Halmos
\end{problem}

The unitary $U_P$ can be compiled into a circuit featuring \CNOT\ gates and \RZ\ gates in $O(2^n)$ layers using standard rewriting rules. However, this is far from optimal for low degree polynomials.

\begin{example}
To fix the ideas, let us consider a simple Hamiltonian, $H_P(Z) = \alpha_{123} Z_1 \otimes Z_2 \otimes Z_3$. The corresponding operator is $U_P = \exp{(-\jmath\alpha_{123}Z_1 \otimes Z_2 \otimes Z_3)}$. A known decomposition~\citep{HOBO,BK12} of this operator using the allowed gates is depicted in Figure \ref{fig:decomp_example}a. 
\Halmos
\end{example}

\begin{example}
Next, we consider a slightly more complicated Hamiltonian of the following form: 
$$
H_P(Z) = \underset{I \subseteq \{1,2,3\}}{\sum} \alpha_I  \underset{i \in I}{\bigotimes} Z_i.
$$ 
Inspired by Gray code, in \cite{HOBO}, the authors propose the circuit shown in Figure \ref{fig:decomp_example}b. This circuit lets us encode a polynomial of degree three. However, it is not clear how one would do for more complex polynomials, for instance the sum of $H_P(Z)$ above with a second polynomial $H_2(Z) = \underset{I \subseteq \{2,3,4\}}{\sum} \beta_I \otimes_{i \in I} Z_i$. Compiling the first and then the second and subsequently putting them together one after the other would include twice all the monomials that are common to the two Ising Hamiltonians, in this case $\{2\}$, $\{3\}$, and $\{2,3\}$. \Halmos
\end{example}

\begin{figure}[b]
\centering
    \begin{subfigure}{0.285\textwidth}
    \centering
    \includegraphics[width=\textwidth]{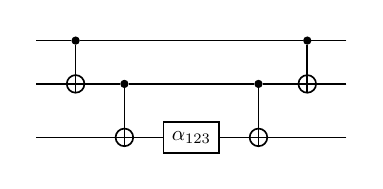}
    \subcaption{Circuit for Example~1.}
    \end{subfigure}
    \hfill 
    \begin{subfigure}{0.695\textwidth}
    \centering
    \includegraphics[width=\textwidth]{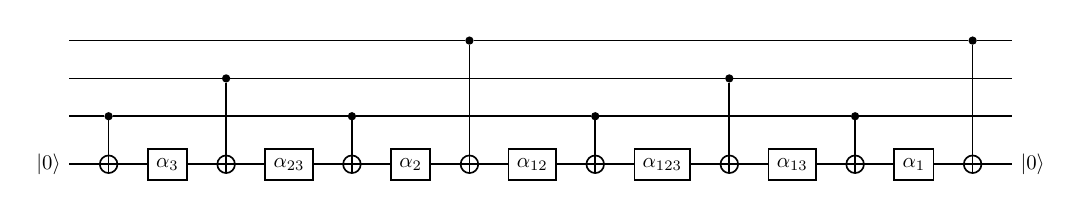}
    \subcaption{Circuit for Example~2.}
    \end{subfigure}
    \caption{Quantum circuits for a few simple Ising Hamiltonians.}
    \label{fig:decomp_example}
\end{figure}

From the discussion of the above examples, it seems reasonable that a better strategy to encode the Hamiltonian should be achievable. In this paper, we take the point of view of \emph{routing}. In fact, from a more high-level abstraction, the circuit in Figure \ref{fig:decomp_example}a can be seen as a form of routing of the singletons $Z_i$, in order to form the product $Z_1 \otimes Z_2\otimes  Z_3$ on the third qubit, at which moment, the appropriate phase rotation gate is applied. 

Following this point of view, at the start of the circuit, each qubit represents its corresponding singleton: $Z_1$ is on qubit 1, $Z_2$ on qubit 2, $Z_3$ on qubit 3, and so forth. The first \CNOT\ gate from qubit 1 to qubit 2 replaces the current monomial on qubit 2 by the product of the monomials on qubit 1 and 2, that is to say $Z_2$ is replaced by $Z_1 \otimes Z_2$. In this context, we say that qubit 1 is {\bf routed} or {\bf sent} to qubit 2 to form the product $Z_1 \otimes Z_2$. 

This interpretation is also valid for the circuits in Figures~\ref{fig:decomp_example}b, where we simply have to consider that the ancilla qubit starts by representing the constant $1$. 

This view leads the search for the shallowest circuit for a given Hamiltonian as a combinatorial optimization problem, where variables and constraints describe which gates are put on which qubits, and this routing process of the monomials through the circuit through the use of the \CNOT\ gates. 


\section{An \NP-complete result for topologically-constrained problems}\label{sec:np-completeness}

We are now ready to provide an \NP-completeness result. First of all, one may think that the circuit design problem is easier when the degree of the polynomial is lower or when we require a small circuit depth. In this section, we show that the complexity of \pbac depends mainly on the number of qubits and monomials and not on the degree or the depth of the circuit. To do so, we give a proof of the \textsf{NP}-Completeness of a variant of \pbac under a topology constraint of the qubit connectivity - that is we are not allowed to connect every pair of qubits. In the proof, the degree of the polynomial and the bound on the depth are constant values fixed respectively to $2$ and $9$. 

To our knowledge, determining the complexity class of (the unconstrained version of) \pbac is open. One variant was previously considered and shown \textsf{NP}-Complete in \citep{Amy2019} where the \(\cnot{}{}\) gates must all target the same ancilla qubit. The connectivity is described by a directed topology graph $\mathcal T$ associating each qubit with a vertex and defining the right to place a $\cnot{i}{j}$ gate with the arc $(i, j)$. For example in \citep{Amy2019}, the topology graph is a star where the arcs are directed toward the center. 

The usual topology graph of \pbac is the complete bidirectional graph, referred to as ``all-to-all topology''. Here, we define the following decision variant of \pbac where $d, T \in \mathbb{N}$.

\begin{problem}[($d, T$)-\arpbac]
Given a Pauli-Z polynomial \(H_P(Z) = \sum_{\mu M} \alpha_\mu \bigotimes_{j=1}^n Z_j^{\mu_j} \) with degree $d$, a 3-regular acyclic directed graph $\mathcal D$ on $n$ vertices, is there a quantum circuit implementing \(U_P = \exp{(-\jmath H_P(Z))}\) of depth at most $T$, using only \CNOT\ and \RZ\ gates and satisfying the topology constraint defined by $\mathcal D$? \Halmos
\end{problem}

\begin{theorem}
    \label{thm:np-completeness:1}
    (2, 9)-\arpbac is \textsf{NP}-Complete. 
\end{theorem}
\begin{proof}
    We offer the proof in Appendix~\ref{app.proof}.
\end{proof}

In the following sections, we focus on the unconstrained version of \pbac, with no topology constraint. 
In the light of the proof for ($d, T$)-\arpbac (which is different from \pbac), we conjecture that \pbac is \textsf{NP}-Complete. The main obstacle for the proof consists mainly in constraining the \CNOT\ gates without any topology constraint
\footnote{The very recent connected work~\citep{acuaviva2026cnotdistance} appeared while we were writing this manuscript. There, the authors prove the \textsf{NP}-Completeness of optimizing the size of \CNOT\ circuits (without $\rz{}{}$ gates and without topology constraint) having as input the $\mathrm{GL}_n(\mathbb{F}_2)$ matrix representing the circuit. A promising perpective would then to mix their proof and ours to reach the goal of proving that \pbac is \textsf{NP}-Complete.}
, using only the bound on the depth, the limited number of qubits and the high number of rotation gates. On the other side, any polynomial algorithm for \pbac, if such an algorithm exists, must be based on the fact that no topology is given: any generic algorithm that can be adapted to any given topology cannot be polynomial unless \(\textsf{P} = \textsf{NP}\). 

\section{A template-pruning model}\label{sec:TPA}

We present now an efficient approach to the compilation of $\pbac$.  The method provides a template circuit for a polynomial of a given degree. The template is generated recursively (that is, the template of degree $d$ is based on the one of degree $d-1$) and it is then combined and simplified (or pruned) to obtain a circuit for a given polynomial. 

\subsection{Template generation}

The original form of the template can be found in our preliminary work \citep{verchere2023optimizing}. We present here a different implementation, to which we add a pruning technique to adapt to the polynomial under consideration. While the aforementioned authors apply the template on each monomial of a polynomial in $\{0,1\}$ variables, and then concatenate all the circuits thus obtained, we apply it to the whole polynomial in $\{-1,1\}$ at once, avoiding duplicates. 

The goal of this template is to be generated quickly, in a recursive manner. It generates a circuit for any binary polynomial instance of $n$ variables,  degree \(d\) containing \(m\) monomials in \(O(2^n)\) layers, which is later reduced to $O(m n)$, in the worst-case, but practically less than $O(m d)$. 

Since the template is recursive, we start with a polynomial of maximum degree $n=d=2$. Let \(P\) be a generic bivariate binary polynomial, written as \(P(z) = \alpha_1 z_1 + \alpha_2 z_2 + \alpha_{12} z_1 z_2\). It is the base case for the template, corresponding to the Hamiltonian
$
H_P(Z) = \alpha_1 Z_1 + \alpha_2 Z_2 + \alpha_{12} Z_1 \otimes Z_2 
$. The corresponding time-evolution unitary \(U_P = \exp{(-\jmath H_P(Z))}\) can be written as
\[
U_P = \left( \exp{(-\jmath \alpha_1 Z_1)} \right) \left( \exp{(-\jmath\alpha_2 Z_2)} \right) \left( \exp{(-\jmath\alpha_{12} Z_1 \otimes Z_2)} \right), 
\]
thanks to the commutativity of the \(Z\) gates, and can be further simplified as 
\[
U_P = \left( \exp{(-\jmath \alpha_1 Z_1)} \right) \left( \exp{(-\jmath \alpha_1 Z_2)}\right) \cnot{1}{2} \left( \exp{(-\jmath \alpha_{12} Z_2)} \right) \cnot{1}{2}. 
\]
The circuit implementing this polynomial, denoted by \(\Te_2\) for template of $n=d=2$, is depicted in Figure~\ref{subfig: C2}. The construction of the circuit \(\Te_{n+1}\) from  the circuit \(\Te_n\) resembles the recursivity of the Gray code, a total ordering of binary words of the same size such that two consecutive words only differs by a bit flip. Because the Hamming distance between each word is minimal and all words belong to the sequence, the Gray code is an optimal way to sequentially visit all words while minimizing the total number of bit flips.

\begin{figure}
\centering
    \begin{subfigure}{0.34\textwidth}
    \centering
    \begin{quantikz}[row sep = {0.7cm,between origins}]
        & \gate{\alpha_1} & \ctrl{1} & \qw                & \ctrl{1} & \qw \\
        & \gate{\alpha_2} & \targ{}  & \gate{\alpha_{12}} & \targ{}  &  \qw 
    \end{quantikz}
    \subcaption{Template circuit \(\Te_2\).}
    \label{subfig: C2}
    \end{subfigure}
    \hspace{0.8cm}
    \begin{subfigure}{0.55\textwidth}
    \centering
    \begin{quantikz}[row sep = {0.7cm,between origins}, column sep = 0.25cm]
        & \gate{\alpha_1} & \ctrl{2} & \qw                & \qw           & \ctrl{1}             & \ctrl{2}           & \ctrl{1}            & \qw      & \qw \\
        & \gate{\alpha_2} & \qw      & \qw                & \ctrl{1}      & \targ{}              & \gate{\alpha_{12}} & \targ{}             & \ctrl{1} & \qw \\
        & \gate{\alpha_3} & \targ{}  & \gate{\alpha_{13}} & \targ{}       & \gate{\alpha_{123}}  & \targ{}            & \gate{\alpha_{23}}   & \targ{}  & \qw
    \end{quantikz}
    \subcaption{Template circuit \(\Te_3\).}
    \label{subfig: C3}
    \end{subfigure}
    \caption{Template circuits for binary polynomials with \(2\) or \(3\) variables. We use the non-standard notation of superposing the \RZ\ gates and the \CNOT\ gate to stress that they can be simultaneously performed.}
    \label{fig:template}
\end{figure}
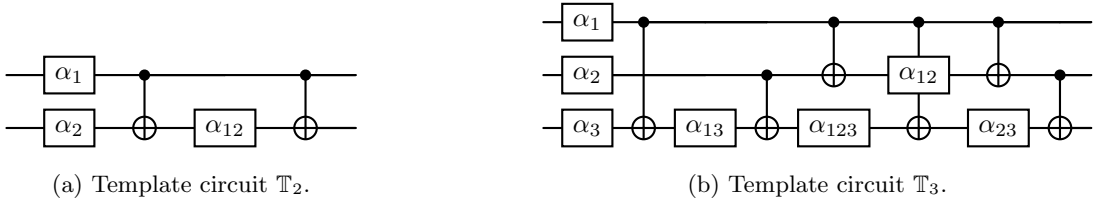

In our context, the binary words represent the monomials, where a letter takes the value \(1\) when the corresponding variable is contained in the monomial. Hence, a bit flip is implement by a \CNOT\ gate. This is precisely what our simplification heuristic focus on, to rewrite the parts of the circuit in which only \CNOT\ gate appears.

The Gray code on \(n + 1\) bits can be constructed from the Gray code on \(n\) bits as follows. We take twice the Gray code on \(n\) bits, and we reverse the order of the second copy of it. Then, we add the letter \(0\) in front of all words of the first copy (the one that is not reversed), and we add the letter \(1\) in front of all words of the second copy. Finally, we concatenate the two copies. 

\begin{example}
The Gray code on two bits is 
$
\{00, \hspace{4pt} 01, \hspace{4pt} 11, \hspace{4pt} 10\}$. We take two copies, we flip the second one, and we get 
$\{
00, \hspace{4pt} 01, \hspace{4pt} 11, \hspace{4pt} 10 \, |\,  10, \hspace{4pt} 11, \hspace{4pt} 01, \hspace{4pt} 00 
\}$. Then, we add a \(0\) in front of the first four words, and a \(1\) to the last four, and concatenate: 
$\{
000, \hspace{4pt} 001, \hspace{4pt} 011, \hspace{4pt} 010, \hspace{4pt} 110, \hspace{4pt} 111, \hspace{4pt} 101, \hspace{4pt} 100 
\}$, which is indeed the Gray code on three bits.  \Halmos
\end{example}

Notice that, by construction, the leftmost bit of the Gray code is always \(1\) on the second half, and always \(0\) in the first half. Hence, when we add a new variable to the polynomial $z_n$, we only have to deal with it in the second half of the Gray code, which is equivalent to the Gray code with one variable less. This corresponds to the partial factorization of our polynomial. More formally,
\begin{align}
P(z) &= \sum_{\mu \in M} \alpha_\mu \prod_{j=1}^n z_j^{\mu_j} \\ 
&=\sum_{\mu \in M} \alpha_\mu \left(\prod_{\substack{j=1, \mu_n = 1}}^{n} z_j^{\mu_j} + \prod_{\substack{j=1, \mu_n = 0}}^{n} z_j^{\mu_j} \right) = z_n \sum_{\substack{\mu \in M, \mu_n = 1}} \alpha_\mu \prod_{j=1}^{n-1} z_j^{\mu_j} + \sum_{\substack{\mu \in M, \mu_n = 0}} \alpha_\mu \prod_{j=1}^{n-1} z_j^{\mu_j}  \\& =: z_n Q(z) + R(z), \label{eq:QR}
\end{align}
where \(Q\) and \(R\) are two generic polynomials on \(n-1\) variables.

If we implement the monomials of \(P\) in the order given by the Gray code, all monomials of \(R\) will come first. Furthermore, constraining the \(\cnot{}{}\) gates to point only downwards, every monomial of \(z_n Q\) needs to be implemented on the last qubit. The template focuses then on implementing the time-evolution operators for monomials of \(z_n Q\) on the last qubit using the Gray code, and then implements \(R\) elsewhere in the circuit. Because this template is recursive, the configuration of the monomials of \(R\) is already determined. 

Formally, we have the algorithm presented in Algorithm~\ref{alg:cap}, where we have indicated with $\llbracket n+1 \rrbracket := \{1, \ldots, n+1\}$. To illustrate this template, we also depicted  the circuit \(\Te_4\) in Figure~\ref{fig: C4}.

\begin{algorithm}
\caption{Template generator}\label{alg:cap}
\small
\begin{algorithmic}[1]
\Require A template for $\Te_n$ in $n$ qubit and $2^n$ layers. 
\Ensure A template for $\Te_{n+1}$ in $n+1$ qubit and $2^{n+1}$ layers.
\State Initialize \(\Te_{n+1}\) as an empty circuit
\State Add all the \RZ\ gates $\alpha_k$, $k \in \llbracket n+1 \rrbracket$ on the first layer \hfill (layer \(1\))
\State Build \(\Te_n^*\) from \(\Te_n\) by removing its first and last layers
\State Add \(\Te_n^*\) to \(\Te_{n+1}\), on qubits \(\llbracket n-1 \rrbracket \cup \{n+1\}\) \hfill (layers \(2 \to 2^n-1\))
\State Build \(\Te_n^\star\) from \(\Te_{n+1}\) by removing the gates not involving qubit \(n+1\)
\State Add a \(\cnot{}{}\) gate to \(\Te_{n+1}\) controlled by qubit $n$ acting on qubit $n+1$ \hfill (layer \(2^n\))
\State Build \(\Te_n^\circ\) from \(\Te^*_n\) by removing the gates not involving the last qubit
\State Add \(\Te_n^\circ\) to \(\Te_{n+1}\) on the qubits \(\llbracket n \rrbracket\) \hfill (layers \(2^n+1 \to 2^{n+1} - 1\))
\State Add \(\Te_n^\star\) to \(\Te_{n+1}\) on \(\llbracket n-1 \rrbracket \cup \{n+1\}\) in reverse order \hfill (layers \(2^n + 1 \to 2^{n+1} - 1\))
\State Add a \(\cnot{}{}\) gate to \(\Te_{n+1}\) controlled by qubit $n$ acting on qubit $q+1$ \hfill (layer \(2^{n+1}\))
\end{algorithmic}
\end{algorithm}

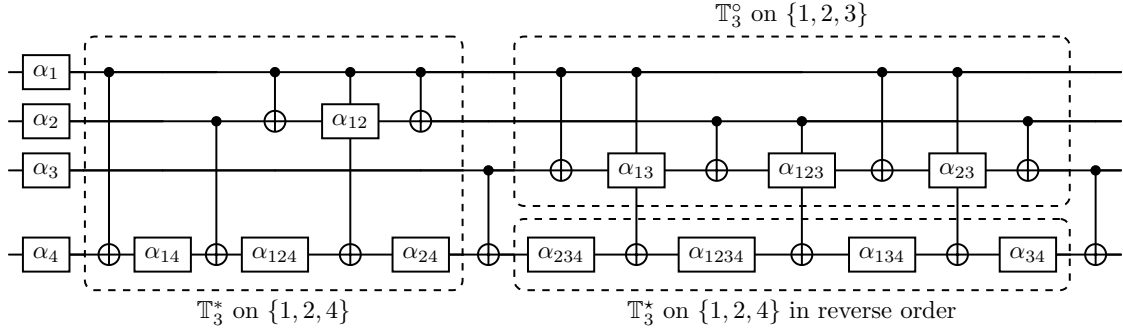
\begin{figure}
    \centering
    \begin{adjustbox}{width=0.9\textwidth}
    \begin{quantikz}[row sep = {0.7cm,between origins}, column sep = 0.2cm]
    & \gate{\alpha_1} & \qw & \ctrl{3}\gategroup[4, steps = 6, style = {dashed, rounded corners, inner xsep = 2pt}, label style = {label position = below, anchor = north, yshift = -0.2cm}]{\begin{minipage}{4cm} \begin{center} \(\Te_3^*\) on \(\{1, 2, 4\}\) \end{center} \end{minipage}} & \qw & \qw & \ctrl{1} & \ctrl{3} & \ctrl{1} & \qw & \qw & \qw & \ctrl{2}\gategroup[3, steps = 7, style = {dashed, rounded corners, inner xsep = 2pt}, label style = {label position = above, anchor = south, yshift = -0.2cm}]{\(\Te_3^\circ\) on \(\{1, 2, 3\}\) } & \ctrl{3} & \qw & \qw & \ctrl{2} & \ctrl{3} & \qw & \qw & \qw &\qw \\
    & \gate{\alpha_2} & \qw      & \qw & \qw & \ctrl{2} & \targ{} & \gate{\alpha_{12}} & \targ{} & \qw & \qw & \qw & \qw & \qw & \ctrl{1} & \ctrl{2} & \qw & \qw & \ctrl{1} & \qw & \qw & \qw \\
    & \gate{\alpha_3} & \qw      & \qw & \qw & \qw & \qw & \qw & \qw & \qw & \ctrl{1} & \qw & \targ{} & \gate{\alpha_{13}} & \targ{} & \gate{\alpha_{123}} & \targ{} & \gate{\alpha_{23}} & \targ{} & \qw &\ctrl{1} & \qw \\[0.5cm]
    & \gate{\alpha_4} & \qw & \targ{} & \gate{\alpha_{14}} & \targ{} & \gate{\alpha_{124}} & \targ{} & \gate{\alpha_{24}} & \qw & \targ{} & \qw & \gate{\alpha_{234}}\gategroup[1, steps = 7, style = {dashed, rounded corners, inner xsep = 2pt}, label style = {label position = below, anchor = north, yshift = -0.2cm}]{\(\Te^\star_3\) on \(\{1, 2, 4\}\) in reverse order } & \targ{} & \gate{\alpha_{1234}} & \targ{} & \gate{\alpha_{134}} & \targ{} & \gate{\alpha_{34}} & \qw & \targ{} & \qw
    \end{quantikz}
    \end{adjustbox}
    \caption{Template circuit \(\Te_4\). Is useful to interpret Algorithm~\ref{alg:cap} as Gray code applied as considering the $i$-th qubit as ancilla for the Gray code of the polynomial up to $z_i$. The last qubit represent $z_4 Q(z)$, while the first three qubits are $R(z)$, see Equation~\eqref{eq:QR}.}
    \label{fig: C4}
\end{figure}

The circuit \(\Te_4\) is built from three \(\Te_3\) circuits, corresponding to the following decomposition: 
\begin{align} 
P(z) = z_n Q(z) + R(z) & = z_n \left( z_{n-1}\!\! \sum_{\substack{\mu \in M, \mu_n = 1, \mu_{n-1} = 1}} \!\!\alpha_\mu \prod_{j=1}^{n-2} z_j^{\mu_j} + \sum_{\substack{\mu \in M, \mu_n = 1, \mu_{n-1} = 0}} \!\!\alpha_\mu \prod_{j=1}^{n-2} z^{\mu_j}\right) + R(z) \\
& = z_n z_{n-1} Q_1(z) + z_n Q_2(z) + R(z), 
\end{align}
where \(Q_1\) and \(Q_2\) are two generic polynomials on \(n-2\) variables. The template starts by implementing the polynomial \(Q_2\), formed of all non-singleton monomials containing \(z_n\) but not \(z_{n-1}\). Once this is done, we apply a \(\cnot{}{}\) gate which bonds \(z_n\) together with \(z_{n-1}\), because \(Q_1\) is the polynomial of all monomials containing \(z_n z_{n-1}\). Because we no longer acts on the (\(n-1\))-th qubit, the template has the space for implementing \(R\) at the same time as \(Q_1\). Notice that the template, due to its similarity with the Gray code, alternate between \(\cnot{}{}\) gate and \(\rz{}\) gate. So we can implement \(R\) and \(Q_1\) at the same time, by simply delaying the implementation of \(R\) by one layer. 

Once we have generated the circuit on the adequate number of qubits using this template, we feed into each rotation gate the coefficient of the corresponding monomial. If a monomial does not appear in the polynomial, we set the angle to zero and remove the $\rz{}{}$ gate. In general, doing so generate inefficient \(\cnot{}{}\) circuits between consecutive non-trivial rotation gates, and these \(\cnot{}{}\) circuits can be simplified using some rewriting rules, as we explain in the Section~\ref{sec:simplification}. 

For Algorithm~\ref{alg:cap}, we have the following theoretical result. 
\begin{proposition}\label{prop:1}
For polynomial $P(z)$ in $n$ variables, the Template generator in Algorithm~\ref{alg:cap} outputs a circuit $\Te_{n}$ given the preceding circuit $\Te_{n-1}$ that has all the rotation gates we need to represent a polynomial in $z$ of degree $n$ and contains all the products of $Z$'s. It has exactly $2^{n}$ layers and $2^{n}$ unique \RZ\ gates.
\end{proposition}

\begin{proof}We offer the proof in Appendix~\ref{ap:prop:1}.
\end{proof}  

\begin{remark}[A walk on the hypercube and other techniques]
It is useful to interpret Algorithm~\ref{alg:cap} as Gray code applied as considering the $i$-th qubit as ancilla for the Gray code of the polynomial up to $z_i$. The Gray code represents an Hamiltonian path on the hypercube formed by all the possible $2^n$ monomials. A way to reduce the depth of the circuit (in the case of less non-zero monomials) would then be to consider optimal paths on sparser graphs, e.g., Steiner trees/paths. This has been investigated in the quantum computing literature for other compilation tasks~
\citep{vandaele2022phase,de2022dynamic,gavrielov2024linear}, but has the disadvantage to be computationally costly. Since the optimal reduction seems out of reach in general~\citep{acuaviva2026cnotdistance}, in the following, we give a heuristic yet effective and scalable pruning technique based on transvection operations. \Halmos
\end{remark}

\subsection{The pruning method} \label{sec:simplification}

In general, the binary polynomial to encode does not contain every monomial of its variables. Then, when using the template described above, around the rotation gates with a null angle, the layout of \CNOT\ gates is not efficient. So, between each pair of consecutive non-trivial rotation gates, we want to rewrite the circuits made of \CNOT\ gates. Given our template, we focus on circuits solely made of {\bf downward} \CNOT\ gates. We will use group theoretic concepts to facilitate our task. 

Each circuit corresponds to an element of the group generated by the \CNOT\ gates, represented by a given word of \CNOTs\ determined by the template. One unique group element can be written using many different words, of different lengths. The point of the following discussion is to present a method that takes as an input a specific word of \CNOT\ gates, for which the index of the control qubit is always smaller than the index of the target qubit, and that return another word of \CNOT, corresponding to the same group element, in a canonical form, and with at most linear depth. 

The methods presented in this part of the paper resemble standard methods used in the literature, in particular by~\citep{maslov2007linear}, \citep{markov2008optimal}, as well as~\citep{gavrielov2024linear}, even though the goals and the implementation differ\footnote{
We also cite~\citep{Webster2026heuristicoptimal,de2021gaussian} which provide alternative algorithms to \CNOT\ synthesis. The former is typically too slow for our purpose, the latter less performant in practice in term of depth for our application. In general, the authors focus on reducing the total number of gates, rather than having a quick heuristic with bounded linear depth, as we do. So we have founded our simplification very effective in practice. 
}
. In particular, we make express use of the downward nature of the \CNOTs, yielding lower-triangular matrices. 

The method contains two phases. Let \(\mathcal{C}\) be a given circuit only composed of \CNOT\ gates, for instance, a subcircuit comprised between two consecutive \RZ\ gates. The first step is to translate \(\mathcal{C}\) into a matrix \(M_\mathcal{C}\). This matrix encodes the action on the quantum states implemented by the circuit, hence, all the circuits of \CNOT\ implementing the same unitary correspond to the same matrix (note, however, that the matrix \(M_\mathcal{C}\) has nothing to do with the unitary matrix that acts on quantum states, it is a much more compact representation of the \CNOT\ gates). The second step is to transform back the matrix \(M_\mathcal{C}\) into a circuit \(\mathcal{C}_\star\), with several properties: \emph{(i)} \(\mathcal{C}_\star\) is canonical, i.e., two different initial circuits being translated into the same matrix are transformed back as the exact same circuit; \emph{(ii)} \(\mathcal{C}_\star\) is of at most quadratic size and linear depth in the number of qubits, independently of the number of gates in the initial circuit; \emph{(iii)} in general, the circuit \(\mathcal{C}_\star\) is much shorter than the original circuit \(\mathcal{C}\).

\paragraph{From the circuit to the matrix.}

We start by presenting the isomorphism\footnote{A bijection \(f\colon (G_1, \cdot) \to (G_2, *)\) preserving the group structure, i.e., \(f(x \cdot y) = f(x) * f(y)\).} allowing us to manipulate \CNOT\ circuits using much more common objects. Denote by \(\mathbb{F}_2\) the field with two elements, by \(\mathrm{GL}_n(\mathbb{F}_2)\) the group of \(n \times n\) invertible matrices whose entries belongs to \(\mathbb{F}_2\), equipped with the usual matrix product. Denote by \({\bf T}^{(ij)}\) the transvection matrix associated with the pair \((i, j)\) satisfying \(i \neq j\), defined by 
\begin{equation}
{\bf T}^{(ij)}_{kl} = \begin{cases}
1, & \quad \text{if } k = l \text{ or } (k, l) = (i, j), \\
0, & \quad \text{otherwise}. 
\end{cases}
\end{equation}
A transvection matrix is the sum of the identity matrix and a matrix composed of zeros everywhere except at the entry \((i, j)\). Transvections generate \(\mathrm{GL}_n(\mathbb{F}_2)\) as a group~\cite[Theorem~8.8]{rotman1999introduction}. 

\begin{theorem}[{\cite[Theorem~6]{bataille2022quantum}}]\label{th:iso-gl}
The group generated by \CNOT\ gates acting on \(n\) qubits is isomorphic to \(\mathrm{GL}_n(\mathbb{F}_2)\). The group homomorphism\footnote{A map between group preserving the group structure, not necessarily bijective.} sending the \CNOT\ gate \(\cnot{j}{i}\) to the transvection \({\bf T}^{(ij)}\) is a group isomorphism. 
\end{theorem}

So, the action of a \CNOT\ circuit is represented as an invertible \(n \times n\) matrix over \(\mathbb{F}_2\). From a given \CNOT\ circuit, it is easy to construct this matrix: we simply take the product of all transvections of the \CNOT\ gate of the circuit, in the reverse order. 

\begin{example} 
Consider the following circuit and its corresponding element of \(\mathrm{GL}_3(\mathbb{F}_2)\). We consider the \CNOT\ gates from left to right, and we take the product of the corresponding transvection from right to left: 
\[
\begin{quantikz}[row sep = {0.7cm,between origins}, column sep = 0.2cm]
    & \qw & \targ{} & \qw & \qw\\
    & \targ{} & \qw & \ctrl{1} & \qw \\
    & \ctrl{-1} & \ctrl{-2} & \targ{} & \qw \\
\end{quantikz}
\qquad 
\raisebox{0.25cm}{%
$\Rightarrow$ \qquad
\begingroup
\renewcommand*{\arraystretch}{0.7}
$A = {\bf T}^{(32)} {\bf T}^{(13)} {\bf T}^{(23)}
= \begin{pmatrix} 1 & 0 & 0 \\ 0 & 1 & 0 \\ 0 & 1 & 1 \end{pmatrix}
\begin{pmatrix} 1 & 0 & 1 \\ 0 & 1 & 0 \\ 0 & 0 & 1 \end{pmatrix}
\begin{pmatrix} 1 & 0 & 0 \\ 0 & 1 & 1 \\ 0 & 0 & 1 \end{pmatrix}
= \begin{pmatrix} 1 & 0 & 1 \\ 0 & 1 & 1 \\ 0 & 1 & 0 \end{pmatrix}.$
\endgroup
}\Halmos
\]
\end{example}

The isomorphism described in Theorem~\ref{th:iso-gl} specializes conveniently when we consider circuits containing only downward \CNOT\ gates. Indeed, the image of the group of circuits with downward \CNOT\ gates is the group of lower-triangular invertible matrices over \(\mathbb{F}_2\), a subgroup of \(\mathrm{GL}_n(\mathbb{F}_2)\). This subgroup is an example of a Borel subgroup of \(\mathrm{GL}_n(\mathbb{F}_2)\), that are well-studied and well-structured, and it is precisely this structure that we leverage in the following. 

\paragraph{From the matrix to the circuit.} 

The task now is to be able to read out a \CNOT\ circuit from an element of \(\mathrm{GL}_n(\mathbb{F}_2)\). In general, this task is much more complicated than in the opposite direction. However, this can be done easily for triangular matrices~\cite[Theorem~4]{laplace2025variational}. Let \((Q, \prec)\) be the ordered set with \(Q = \{(i, j) \mid 1 \leq j < i \leq n\}\), equipped with the anti-lexicographic order \(\prec\) given by 
$
(i_1, j_1) \prec (i_2, j_2)$ if and only if $[i_1 > i_2] \hspace{6pt} \text{or} \hspace{6pt} [i_1 = i_2 \text{ and } j_1 > j_2]
$. This is a total order on \(Q\). Denote by \(w_\star\) the word obtained by multiplying, on the left, the transvections indexed by \(Q\) in the order of \(\prec\):
$
w_\star = {\bf T}^{(21)} {\bf T}^{(31)} {\bf T}^{(32)} {\bf T}^{(41)} \cdots {\bf T}^{(n-1 \hspace{2pt} n-2)} {\bf T}^{(n1)} \cdots {\bf T}^{(n \hspace{2pt} n-2)} {\bf T}^{(n \hspace{2pt} n-1)}\,.
$

\begin{theorem}[{\cite[Theorem~4]{laplace2025variational}}]
    Any lower-triangular matrix in \(\mathrm{GL}_n(\mathbb{F}_2)\) can be written as a subword of \(w_\star\).
\end{theorem}

Let \(M\) be an lower-triangular element of \(\mathrm{GL}_n(\mathbb{F}_2)\). Using the theorem above, we have that there exists a word \(w_M \le w_\star\) that represents this element \(M\). Thanks to the specific ordering in \(Q\), this subword is exactly the ordered product of the elements of the following set: 
\begin{equation}\label{eq:q-m}
Q_M = \left\{ {\bf T}^{(ij)} \mid M_{ij} = 1, \text{ for all } i > j \right\}.
\end{equation}
To get the corresponding circuit, we read the word from right to left, and for each transvection, we add the corresponding \CNOT\ gate. The whole pruning procedure is described in Algorithm~\ref{alg:simplification}.

\begin{algorithm}
\caption{Pruning procedure}\label{alg:simplification}
\small
\begin{algorithmic}[1]
\Require A circuit \(C\) only containing \(\cnot{}{}\) and \(\rz{}\) gates
\Ensure Another circuit \(C'\) only containing \(\cnot{}{}\) and \(\rz{}\) gates
\Function{Simplify}{$C$}
    \State \(c \gets \) empty circuit
    \For{$0 \leq \ell \leq$ depth of $C$} 
        \If{layer $\ell$ only contains $\cnot{}{}$ gates}
            \State add layer \(\ell\) to \(c\)
        \Else{}
            \State add layer \(\ell\) to \(c\) without the \(\rz{}\) gate(s)
            \State \(M \gets\) the identity matrix
            \For{$0 \leq k \leq$ depth of $c$}
                \For{each $\cnot{i}{j}$ in layer $k$ of $c$}
                    \State \(M \gets {\bf T}^{(ji)} M\)
                \EndFor
            \EndFor
            \State \(c' \gets\) empty circuit
            \For{$(i, j)$ in $Q_M$ read in the order given by \(\prec\)} \Comment{see Equation~\eqref{eq:q-m}}
                \State add gate \(\cnot{j}{i}\) to \(c'\)
            \EndFor
            \State concatenate \(c'\) to the right of \(C'\)
            \State add to \(C'\) the \(\rz{}\) gates of layer \(\ell\) of \(C\)
            \State \(c \gets\) empty circuit
        \EndIf
    \EndFor
    \State \textbf{return} \(C'\)
\EndFunction
\end{algorithmic}
\end{algorithm}

To summarize, we have seen for now that the group of \(\cnot{}{}\) circuits is isomorphic to \(\mathrm{GL}_n(\mathbb{F}_2)\), and the \(\cnot{}{}\) circuits extracted from the template circuits described previously form the subgroup of lower-triangular matrices. Then, each lower-triangular element of \(\mathrm{GL}_n(\mathbb{F}_2)\) can be written as a {\bf linear depth}, quadratic size circuit of \CNOT\ gates (see~\cite{laplace2025variational}).

\subsection{Putting things together: \VT}

The algorithm \VT\ (the name comes from a letter shuffling based on the lengthy template-pruning algorithm), reported in Algorithm~\ref{alg:VT} describes the full template-pruning algorithm. 
in Appendix~\ref{app.examples}), we offer some examples of how \VT\ behaves in simple cases.  For \VT, we have the following depth worst-case result. 

\begin{algorithm}
\caption{\VT}\label{alg:VT}
\small
\begin{algorithmic}[1]
\Require An instance $P$ of HUBO with $n$ variables, degree $d$, with $m$ monomials. 
\Ensure A circuit $C'$ of worst-case depth $O(m n)$
\State Build the template $\Te_n$ using Algorithm~\ref{alg:cap} and remove from it the rotation gates corresponding to monomials not in $P$. Obtain a circuit $C$
\State $C'$= Simplify($C$) \Comment{defined in Algorithm~\ref{alg:simplification}}
\end{algorithmic}
\end{algorithm}

\begin{proposition}\label{prop.2}
For polynomial $P(z)$ in $n$ variables, $m$ non-zero monomials, of degree $d$, the algorithm defined as \emph{\VT}\ outputs a circuit that has all the rotation gates we need to represent $P(z)$. It has at most $O(m n)$ layers.
\end{proposition}

\begin{proof} We offer the proof in Appendix~\ref{ap:prop:2}. 
\end{proof}

Proposition~\ref{prop.2} offers a worst-case depth bound for \VT. The bound is really a worst-case in $n$, but rather tight in $m$, due to the Gray code construction. In fact, a naive implementation of a Gray code, we can achieve $\sum_{\mu \in M} 2 \lvert \mu \rvert - 1 \leq 2 m d = O(md)$ easily by arbitrarily ordering the monomials, and, for each monomial $\mu$, add $|\mu|-1$ \CNOT\ gates to it, put $\mu$ on the qbits having max degree and resetting the qbits with $|\mu|-1$ \CNOT\ gates in the reversed order. We see that in practice, \VT\ is better than the naive $O(m d)$. 
These bounds are to be compared with the $m 2^d$ of~\cite{verchere2023optimizing}, which is rather tight in both $m$ and $d$, giving us an important improvement. To further leverage parallelism and reduce the $m$ dependency, we turn now to binary linear programs.

\subsection{An efficient diagonal ansatz and more}

\VT\ is an efficient template-pruning method to compile binary polynomials for the use in QAOA. But this is not all. The fact that the circuit is parametric in the $\alpha$'s leaves us the opportunity to do a bit more. In particular, we highlight three complementary results. 

First, \VT\ offers an efficient ansatz for diagonal operators in the context of approximate compilation~\citep{madden2022best,zylberman2025efficient}. This means that $\alpha$'s in the circuit can be optimized to match a given target diagonal unitary matrix. One can decide to use the full template (matching the depth bound of $2^n$) to obtain an exact compilation. Or decide to use monomials up to degree $d$, obtaining an approximate compilation in the spirit of~\cite{welch2014efficient} and \cite{zylberman2025efficient}. 

Second, \VT\ can be readily used in multi-angle QAOA~\citep{herrman2022multi}, where all the angles are optimized, and not only a $\gamma$ per layer. 

Third, a \VT\ truncated at depth $d=2$ (or higher) can be readily used in approximate quadratization QAOA~\citep{druagoi2026approximate}, where the ansatz is fully optimized and it has a mismatch with the cost.

\section{Binary Linear Programming models} \label{sec:ILP}

We now formulate $\pbac$ as a binary linear program. Let \(n\) be the number of variables in the problem and \(M\) be the set of monomials ($m = |M|$). The maximum circuit depth considered is denoted by \(T\). This parameter is set as the minimum between the depth given by the template \VT \ described above, and twice the sum of the degrees of the monomials minus 1, i.e., \(\sum_{\mu \in M} 2 \lvert \mu \rvert - 1 = O(md)\), which represents the naive implementation of all monomials sequentially. Crucially, in the worst-case, this depth is linear in both $m$ and $d$. For handy notation we set \(\llbracket \ell, m \rrbracket \coloneqq \{\ell, \ldots, m\}\) and \(\llbracket m \rrbracket \coloneqq \llbracket 1, m \rrbracket\) for any integers \(1 \leq \ell \leq m\). 

Then, we can formulate $\pbac$ as the following binary linear program. 

\footnotesize
\begin{subnumcases}{\CDP \label{mod:milp}}
    \min ~ J(a) \coloneqq \sum_{r \in \llbracket T \rrbracket} a^r  \label{eq:obj}\\
     \qquad a^{r-1} \geq a^r, & $\forall~r \in \llbracket T \rrbracket \setminus 1. $  \label{eq:sym}\\ 
   \qquad b^0_{k, k}  = 1, \, b^T_{k, k}  = 1, &  $\forall~k \in \llbracket n \rrbracket, $ \label{eq:init1}\\
     \qquad b^0_{k, p}  = 0, \, b^T_{k, p}  = 0,& $\forall~k, p \in \llbracket n \rrbracket$ with $k \neq p.$ \label{eq:init2}\\ 
    \qquad \sum_{\ell \in \llbracket n \rrbracket} \left( c^r_{k, \ell} + c^r_{\ell, k} \right) + \sum_{\mu \in M} g^r_{k, \mu} \leq a^r, & $\forall~r \in \llbracket T \rrbracket$ $\forall~k \in \llbracket n \rrbracket.$ \label{eq:activity}\\ 
    \qquad c^r_{k, k} = 0, & $\forall~r \in \llbracket T \rrbracket$ and $\forall~k \in \llbracket n \rrbracket. $ \label{eq:noauto}\\
    \qquad \sum_{r \in \llbracket T \rrbracket} \sum_{k \in \llbracket n \rrbracket} g^r_{k, \mu} = 1, & $\forall~\mu \in M$. \label{eq:compute}\\
    \qquad g^r_{k, \mu} \leq b^{r-1}_{k, p}, & $\forall~p \in \mu$,  $\forall~r \in \llbracket T \rrbracket$, $\forall~k \in \llbracket n \rrbracket$ and $\forall~ \mu \in M$. \label{eq:valid1}\\
    \qquad g^r_{k, \mu} \leq 1 - b^{r-1}_{k, p},  & $\forall~p \in \llbracket n \rrbracket \setminus \mu$,  $\forall~r \in \llbracket T \rrbracket$, $\forall~k \in \llbracket n \rrbracket$ and $\forall~ \mu \in M$. \label{eq:valid2}\\
    \qquad --------------- \nonumber\\
    \qquad b^r_{k, p} \geq  b^{r-1}_{k, p} - \Gamma^r_{k,p}    & $\forall~r \in \llbracket T \rrbracket$, $\forall~k,p \in \llbracket n \rrbracket$ \label{eq:propagation1}\\
    \qquad b^r_{k, p}\geq   \Gamma^r_{k,p} - b^{r-1}_{k, p}   & $\forall~r \in \llbracket T \rrbracket$, $\forall~k,p \in \llbracket n \rrbracket$ \label{eq:propagation2} \\
    \qquad b^r_{k, p}\leq   b^{r-1}_{k, p} + \Gamma^r_{k,p}   & $\forall~r \in \llbracket T \rrbracket$, $\forall~k,p \in \llbracket n \rrbracket$ \label{eq:propagation3} \\
    \qquad b^r_{k, p} \leq 2 - b^{r-1}_{k, p} - \Gamma^r_{k,p}   & $\forall~r \in \llbracket T \rrbracket$, $\forall~k,p \in \llbracket n \rrbracket$ \label{eq:propagation4} \\
        \qquad --------------- \nonumber\\
    \qquad \Gamma^r_{k,p} = \sum_{\ell \in \llbracket n \rrbracket} \gamma^r_{\ell, k,p} & $\forall~r \in \llbracket T \rrbracket$, $\forall~k,p \in \llbracket n \rrbracket$ \label{eq:lin0} \\
        \qquad --------------- \nonumber\\
    \qquad \gamma^r_{\ell, k,p} \leq c^r_{\ell, k} & $\forall~r \in \llbracket T \rrbracket$, $\forall~k,p,\ell \in \llbracket n \rrbracket$ \label{eq:lin1} \\
    \qquad \gamma^r_{\ell, k,p} \leq b^{r-1}_{\ell, p} & $\forall~r \in \llbracket T \rrbracket$, $\forall~k,p,\ell \in \llbracket n \rrbracket$ \label{eq:lin2} \\
    \qquad \gamma^r_{\ell, k,p} \geq c^r_{\ell, k} + b^{r-1}_{\ell, p} -1 & $\forall~r \in \llbracket T \rrbracket$, $\forall~k,p,\ell \in \llbracket n \rrbracket$ \label{eq:lin3} \\
    \qquad --------------- \nonumber\\
    \qquad a,~c,~b,~g,~\gamma,~\Gamma & binary variables \nonumber
\end{subnumcases}
\normalsize
\smallskip

We have introduced the following supporting binary variables taking values in \(\{0, 1\}\):
\begin{itemize}
\item \(a^r\): \(1\) if and only if  something happens in the circuit at layer \(r\), 
\item \(c^r_{k, \ell}\): \(1\) if and only if  there is a \(\cnot{}{}\) gate from qubit \(k\) to qubit \(\ell\) on layer \(r\), 
\item \(g^r_{k, \mu}\): \(1\) if and only if  there is a gate corresponding to monomial \(\mu\) on qubit \(k\) on layer \(r\), 
\item \(b^r_{k, p}\): \(1\) if and only if  the \(p\)-th variable \textbf{is taken into account on qubit }\(k\) (via sequences of \(\cnot{}{}\) gates) between layer \(r\) and $r+1$, after applying the gates on layer $r$. For notation convenience  $b^0_{k, p}$ represents the $\rz{}{}$ applied at the beginning of the circuit. 
\end{itemize}

The objective function~\eqref{eq:obj} aims to minimize the sum of active layers. 
Constraints~\eqref{eq:sym} are symmetry-breaking inequalities that will avoid empty layers unless all the monomials have been taken into account. 
Constraints~\eqref{eq:init1} and~\eqref{eq:init2} state the initial and final conditions of the circuit. On the first layer, each qubit carries only monomials of degree one, corresponding to the variable of the same index. 
On the last layer, uncomputing is required, as this circuit will be part of a larger circuit. 
Constraints~\eqref{eq:activity} model the fact that a layer is active if there is a \(\cnot{}{}\) or a rotation gate on that layer, and a qubit cannot be acted on by more than one gate per layer. 
The fact that a qubit cannot control itself through the use of a \(\cnot{}{}\) gate is enforced by Constraints~\eqref{eq:noauto}. 
With Constraints~\eqref{eq:compute} all the monomials are implemented exactly once.
Constraints~\eqref{eq:valid1} and~\eqref{eq:valid2} enforce rotation validity. A rotation corresponding to a given monomial can only be implemented if that monomial is properly represented on a given qubit.

The following constraints~\eqref{eq:propagation1}-\eqref{eq:lin3} implement the propagation rule. This rule can be written as the non-linear Equation~\eqref{eq: propagation} in the initial variables. The behavior of \CNOT\ gates, and the overall propagation of terms throughout the circuit, can be  explained as a binary sum \(\oplus_2\): 
\begin{equation}\label{eq: propagation}
b^r_{k, p} = b^{r-1}_{k, p} \oplus_2 \sum_{\ell \in \llbracket n \rrbracket} c^r_{\ell, k} b^{r-1}_{\ell, p}, \quad \text{for all layers } r \in \llbracket T \rrbracket, \text{ and qubits and variables } k, p \in \llbracket n \rrbracket.
\end{equation}
which has been translated  in~\cite{verchere2023optimizing} as the following equation:
\begin{equation}\label{eq: propagation2}
b^r_{k, p} = b^{r-1}_{k, p} + \sum_{\ell \in \llbracket n \rrbracket} c^r_{\ell, k} b^{r-1}_{\ell, p} - 2 b^{r-1}_{k, p} \sum_{\ell \in \llbracket n \rrbracket} c^r_{\ell, k} b^{r-1}_{\ell, p}.
\end{equation}

In this paper, we adopt a different writing that will lead to a better computational behavior
\begin{equation}\label{eq: propagation3}
b^r_{k, p} = \left( b^{r-1}_{k, p} - \sum_{\ell \in \llbracket n \rrbracket} c^r_{\ell, k} b^{r-1}_{\ell, p}  \right)^2 \,.
\end{equation}
Equation~\eqref{eq: propagation3} is a quartic polynomial expression. In order to linearize it, we introduce the binary variables:
\begin{itemize}
    \item $\gamma$ that will satisfy $\gamma^r_{\ell, k,p} = c^r_{\ell, k}  b^{r-1}_{\ell, p}$ and this last equality is enforced through Constraints~\eqref{eq:lin1}-\eqref{eq:lin3}
    \item $\Gamma$ that will satisfy Equation~\eqref{eq:lin0}. Observe that, from Constraints~\eqref{eq:activity}, at most one of the $\gamma$ variables in this summation is equal to one so $\Gamma^r_{k,p}$ is also a binary variable.
\end{itemize}
Finally, Equation~\eqref{eq: propagation3} is enforced by the linear equations~\eqref{eq:propagation1}-\eqref{eq:propagation3}.

\begin{proposition}
\CDP\ features \(O(T n^3 + T n m)\) binary variables and \(O(T n^3+ T n^2 m)\) constraints. Since $T$ is at worst $O(m d)$, both counts are polynomial in the input.  
\end{proposition}

\begin{proof}
By direct computation. 
\end{proof}

As one can expect, the problem \CDP \ (depicted in the model~\eqref{mod:milp}) is difficult to solve to optimality in reasonable amounts of time. To tackle this problem, we add some additional constraints to restrict the set of solutions to explore.

\subsection{Downward \CNOT\ gates (\TB)}

To simplify the routing of the variables throughout the circuit, we can ask for only having \CNOT\ gates such that the index of the control qubit is smaller than the index of the target qubits: 
\begin{equation}\label{eq: downward}
c^r_{\ell, k} = 0, \quad \text{for all layers } r \in \llbracket T \rrbracket \text{ and for all pairs of qubits } \ell, k \in \llbracket n \rrbracket \text{ such that } k > \ell. 
\end{equation}

\subsection{Downward \CNOT\ gates with fixed rotation gates (\TBFM)}

To further simplify, we set the rotation gates placement to be the same as in the template described at the previous section, while keeping the new constraints on downward \CNOT\ gates (Equation~\eqref{eq: downward}). The upper bound \(T\) is set to be the depth of the circuit constructed by the template, and the constraint about empty layers (Equation~\eqref{eq:sym}) is dropped. We define a map \(\tau\) as follows: 
\[ \begin{aligned} 
\tau\colon \llbracket T \rrbracket \times \llbracket n \rrbracket \times M & \longrightarrow \{0, 1\} \\
(r, k, \mu) & \longmapsto \tau(r, k, \mu) = \begin{cases}
1, & \text{if the template puts a rotation gate at layer } r \\[-0.15cm] & \text{ on qubit } k \text{ corresponding to the monomial } \mu, \\
0, & \text{otherwise}. 
\end{cases}
\end{aligned} \] 
We have the additional new constraint: 
\begin{equation}
g^r_{k, \mu} = \tau(r, k, \mu), \quad \text{for all layers } r \in \llbracket T \rrbracket, \text{ all qubits } k \in \llbracket n \rrbracket \text{ and all monomials } \mu \in M. 
\end{equation}
Note that the objective is not affected by this change, since only active layers are considered there. 

\section{Performance evaluation}\label{sec:results}

We evaluate the different compiling methods over random and low autocorrelation binary sequences (LABS) instances. In both cases, we compare the following algorithms: Qiskit transpilation algorithm (version 2.3.0), \texttt{tket} transpilation algorithm (version 2.18.1), \VT, a version of \VT\ to which we add a Qiskit pass after the algorithm and we call Q-\VT, as well as the binary linear programs \CDP, \TB, \TBFM. We give a $10$ minute time-out constraint, and if reached Gurobi was asked to yield the best solution found or nothing if it had not found one. 

We use Gurobi 13.0.1 for solving the binary linear program. We report runtime when the code is run on a MacBook Pro 2020, 2GHz quad-core Intel Core i5 with 32GB RAM. The depth of our methods is computed by importing the resulting circuit in Qiskit. 

\subsection{Instance description}

{\bf Random instances.} We generated $130$ random polynomials with $n$ variables, degree $d$, and sparsity $\delta \in[0,1]$, for the expected number of monomials $\lvert M \rvert$ is $\delta \cdot \sum_{k = 1}^d \binom{n}{d}$ (that is a fraction $\delta$ of the set of all possible monomials). To do so, for each possible monomial, that monomial is added with probability $\delta$ to the polynomial. In order to ensure the degree $d$, at least one random monomial with $d$ variables is added to the polynomial if none was added during the process. The coefficients of the monomials are random integers uniformly selected from \(-100\) to \(100\).

{\bf LABS instances.} We also generated a second set of $28$ LABS polynomials\footnote{Given two numbers $n$ and $r$ with $r\leq n$,  the LABS problem consists in finding a sequence of lengths $n$ composed of $-1$ and $1$ that solves the following optimization problem:
$$\min_{s \in \{-1,1\}^n} E_{n,r}(s)=\underset{s \in \{-1,1\}^n}{\min} \sum_{p = 1}^{(r-1)}
\sum_{i =1}^{(n-r+1)}  \left( \sum_{j= i}^{(i+r-1-p)}s_j s_{j+p}\right)^2,$$
with $r\le n$, where $s$ is the sequence and $s_i$ is its $i^{th}$ component. Here we consider the limited reach version where the reach $r$ limits the distance between components that interact. The objective function is named the energy function of the sequence $s$. This optimization problem is by nature a quadric polynomial in $\{-1, 1\}$ variables and is very well studied, see e.g.~\cite{clausen2024solving}.} with fixed degree $4$ and monomials of degree $2$ and $4$. 

\subsection{Experimental results}

All the numerical results are given as tables in Appendix~\ref{app.results}). Here we report summarizing figures.

{\bf Random instances. } We start with the fraction of instances solved below a given depth as well as the runtime in seconds with the respect to the number of monomials (Figure~\ref{fig.1}). We also report the obtained depth with respect to the number of monomials (Figure~\ref{fig.2}). 

As we notice from the figures and the accompanying tables, if we consider the returned depth, the binary linear programs offer the best solutions when they are solved at optimality. \VT\ and Q-\VT\ are in par with Qiskit and \texttt{tket} on small instances, but then they become very competitive in medium-to-large instances. In particular, Qiskit and then \texttt{tket} are deliver the longest circuits. As for runtime, Qiskit is globally the fastest, but \VT\ and Q-\VT\ become the methods of choice in large instances. Notice, in the tables, that the depth of \VT\ is consistently lower than the naive $2 m d$, as discussed.  

{\bf LABS instances. } We report similar results for LABS (Figures~\ref{fig.3}-\ref{fig.4}), yielding globally the same conclusions.

\begin{figure}
\centering
    \includegraphics[width=0.425\textwidth]{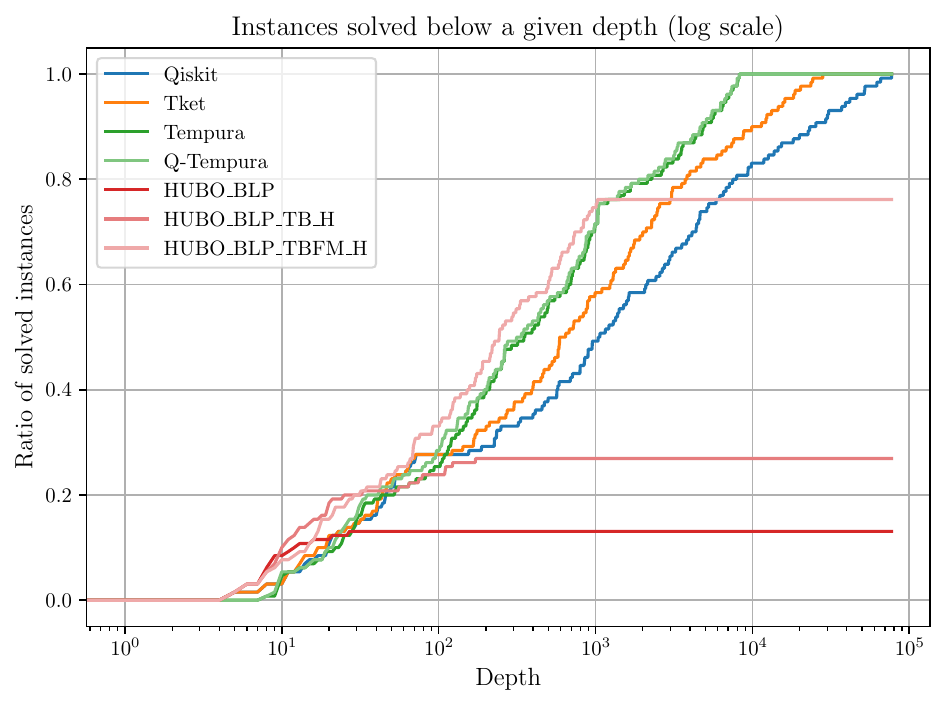}
    \includegraphics[width=0.515\textwidth]{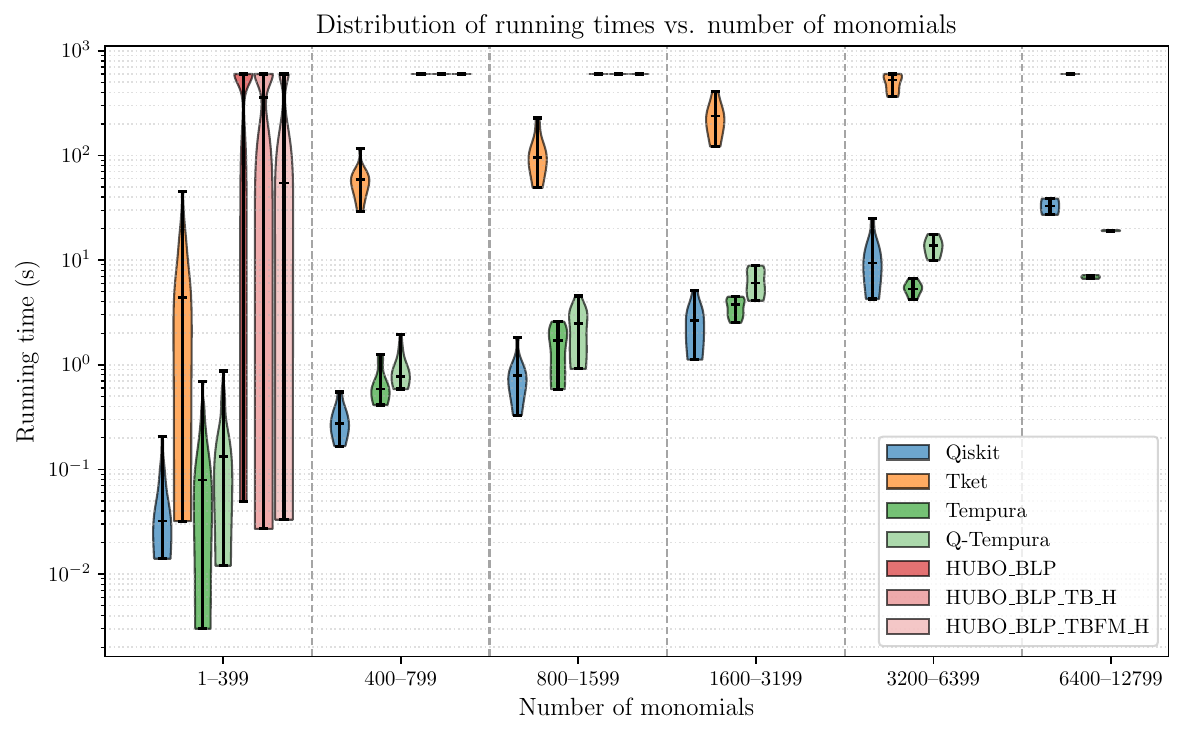}
    \caption{Performance of each method (Random instances). [{\bf Left}] The plots represent the fraction of instances (\(y\)-axis) for which a specific algorithm synthesizes a circuit under a given depth (\(x\)-axis). [{\bf Right}] The plots represent the runtime in seconds for each algorithm versus the number of monomials in the polynomial (\(x\)-axis).}
    \label{fig.1}
\end{figure}
\begin{figure}
\centering
    \includegraphics[width=0.9\textwidth]{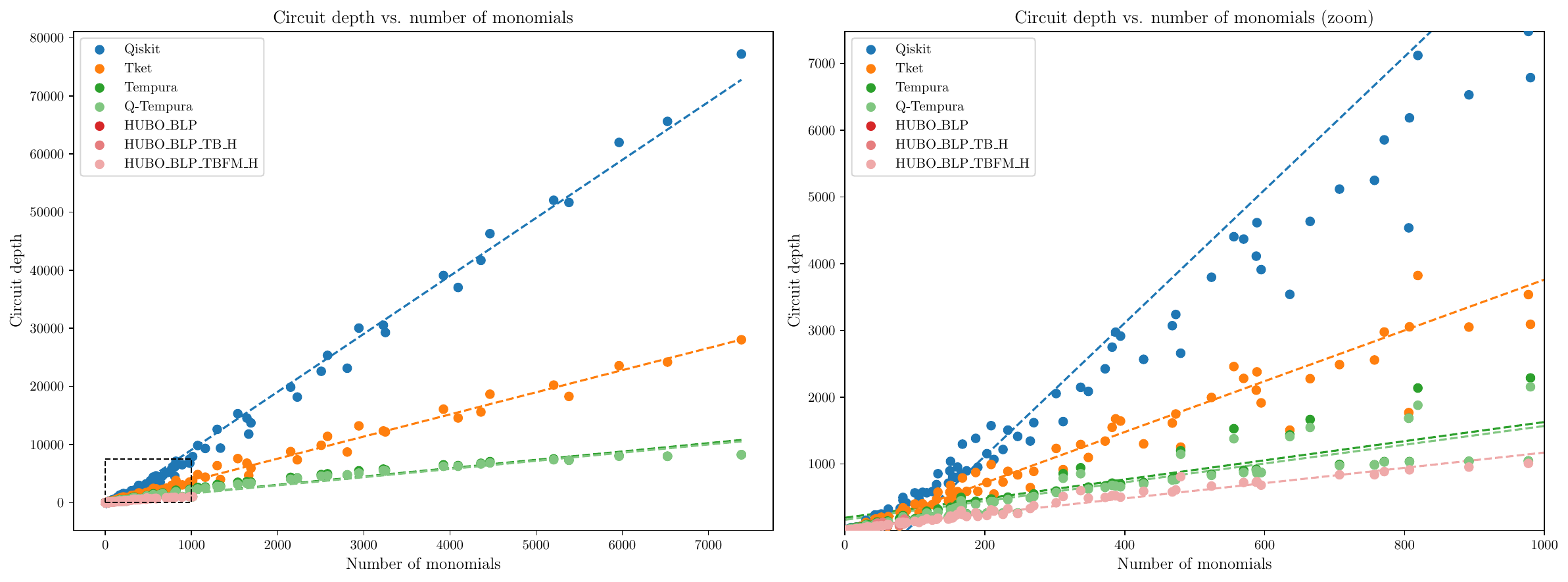}
    \caption{Performance of each method (Random instances). The plots represent the depth (\(y\)-axis) which a specific algorithm yield versus the number of monomials in the polynomial (\(x\)-axis). The regression lines are only for display purpose. }
    \label{fig.2}
\end{figure}

\begin{figure}
\centering
    \includegraphics[width=0.425\textwidth]{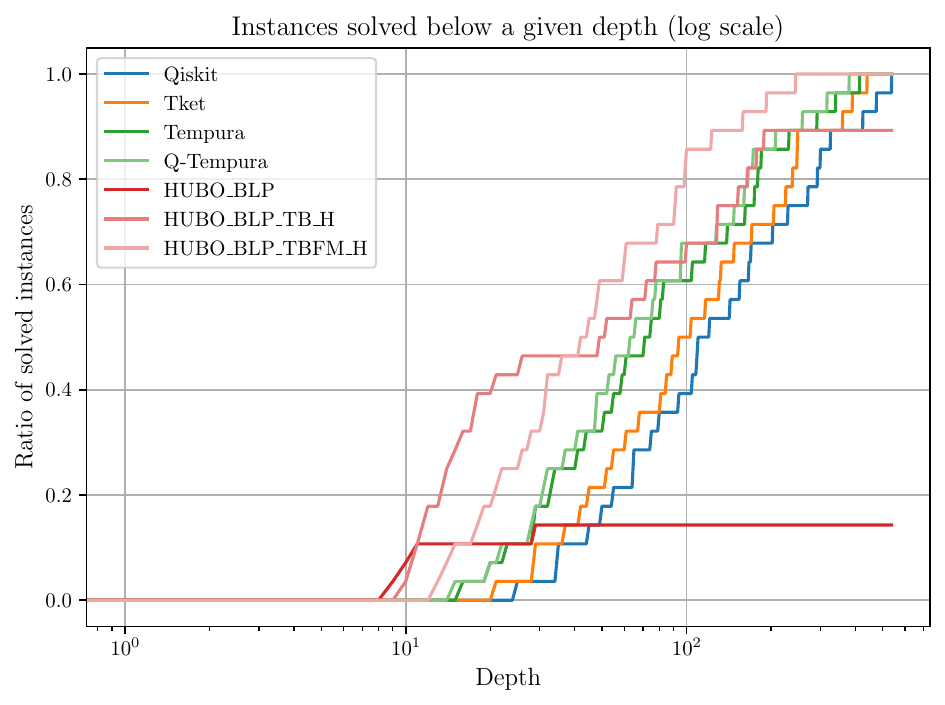}
    \includegraphics[width=0.515\textwidth]{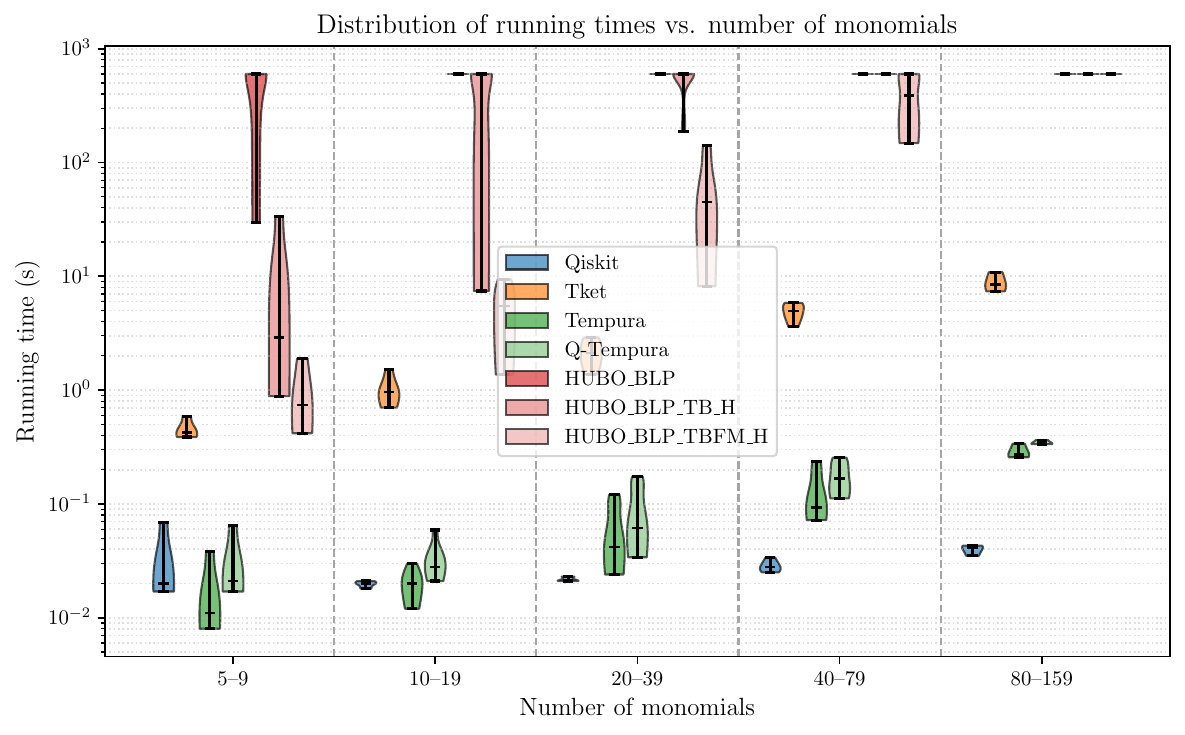}
    \caption{Performance of each method (LABS instances). [{\bf Left}] The plots represent the fraction of instances (\(y\)-axis) for which a specific algorithm synthesizes a circuit under a given depth (\(x\)-axis). [{\bf Right}] The plots represent the runtime in seconds for each algorithm versus the number of monomials in the polynomial (\(x\)-axis).}
    \label{fig.3}
\end{figure}
\begin{figure}
\centering
    \includegraphics[width=0.9\textwidth]{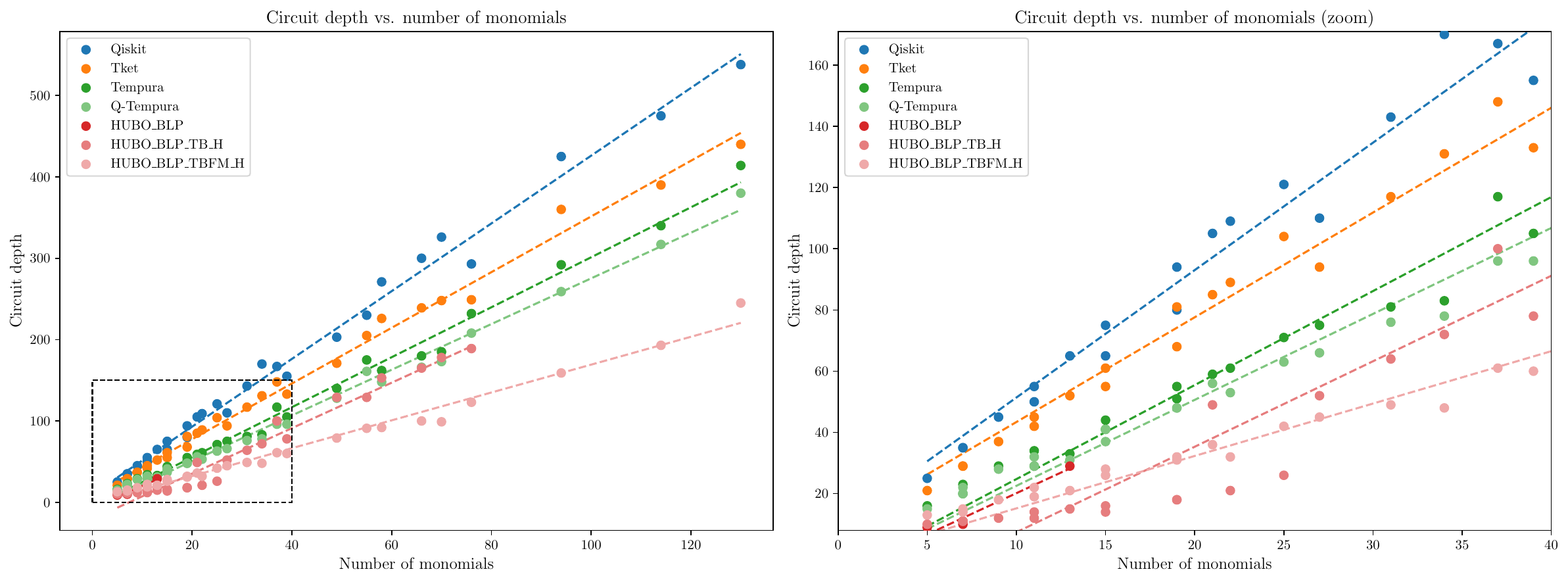}
    \caption{Performance of each method (LABS instances). The plots represent the depth (\(y\)-axis) which a specific algorithm yield versus the number of monomials in the polynomial (\(x\)-axis). The regression lines are only for display purpose.}
    \label{fig.4}
\end{figure}

\section{Conclusions}\label{sec:results}

We have formulated the circuit design problem for binary polynomial optimization as a binary linear program. We have proposed a fast and scalable heuristic solution to it, as well as some variants of the binary programs. We have argue the difficulty of solving such problem, by showing the \NP-completeness of a variant of it. We have supported our algorithms by numerical evaluations with respect the state of the art, yielding far shorter circuits in comparable runtimes (for the heuristic). Given the obtained results, while our heuristic delivers shorter circuits with respect to the state of the art, there is still a considerable gap with the optimal solution, leaving plenty of room for future improvements.

\bibliographystyle{informs2014} 
\bibliography{biblio.bib, Liambib.bib} 

\begin{thebibliography}{45}
\providecommand{\natexlab}[1]{#1}
\providecommand{\url}[1]{\texttt{#1}}
\providecommand{\urlprefix}{URL }

\bibitem[{Acuaviva et~al.(2026)Acuaviva, Acuaviva, \protect\BIBand{}
  Acuaviva}]{acuaviva2026cnotdistance}
Acuaviva A, Acuaviva A, Acuaviva P (2026) {CNOT-Distance is NP-complete under
  all-to-all connectivity}. \urlprefix\url{https://arxiv.org/abs/2608.03825}.

\bibitem[{Amy et~al.(2018)Amy, Azimzadeh, \protect\BIBand{} Mosca}]{Amy2019}
Amy M, Azimzadeh P, Mosca M (2018) On the controlled-not complexity of
  controlled-not–phase circuits. \emph{Quantum Science and Technology}
  4(1):015002, \urlprefix\url{http://dx.doi.org/10.1088/2058-9565/aad8ca}.

\bibitem[{Arora \protect\BIBand{} Barak(2009)}]{arora2009computational}
Arora S, Barak B (2009) \emph{Computational complexity: a modern approach}
  (Cambridge University Press).

\bibitem[{Bataille(2022)}]{bataille2022quantum}
Bataille M (2022) Quantum circuits of cnot gates: optimization and
  entanglement. \emph{Quantum Information Processing} 21(7):269.

\bibitem[{Botea et~al.(2018)Botea, Kishimoto, \protect\BIBand{}
  Marinescu}]{botea2018complexity}
Botea A, Kishimoto A, Marinescu R (2018) On the complexity of quantum circuit
  compilation. \emph{Proceedings of the International Symposium on
  Combinatorial Search}, volume~9, 138--142.

\bibitem[{Clausen et~al.(2024)Clausen, Crama, Lusby, Rodr{\'\i}guez-Heck,
  \protect\BIBand{} Ropke}]{clausen2024solving}
Clausen JV, Crama Y, Lusby R, Rodr{\'\i}guez-Heck E, Ropke S (2024) Solving
  unconstrained binary polynomial programs with limited reach: Application to
  low autocorrelation binary sequences. \emph{Computers \& Operations Research}
  165:106586.

\bibitem[{De~Brugi{\`e}re et~al.(2021)De~Brugi{\`e}re, Baboulin, Valiron,
  Martiel, \protect\BIBand{} Allouche}]{de2021gaussian}
De~Brugi{\`e}re TG, Baboulin M, Valiron B, Martiel S, Allouche C (2021)
  Gaussian elimination versus greedy methods for the synthesis of linear
  reversible circuits. \emph{ACM Transactions on Quantum Computing} 2(3):1--26.

\bibitem[{Dr{\u{a}}goi et~al.(2026)Dr{\u{a}}goi, Baiardi, \protect\BIBand{}
  Egger}]{druagoi2026approximate}
Dr{\u{a}}goi S, Baiardi A, Egger DJ (2026) Approximate quadratization of
  high-order hamiltonians for combinatorial quantum optimization.
  \emph{Physical Review Research} 8(2):023159.

\bibitem[{Farhi et~al.(2014)Farhi, Goldstone, \protect\BIBand{} Gutmann}]{QAOA}
Farhi E, Goldstone J, Gutmann S (2014) A quantum approximate optimization
  algorithm. \emph{arXiv preprint arXiv:1411.4028} .

\bibitem[{Gavrielov et~al.(2024)Gavrielov, Ivrii, \protect\BIBand{}
  Garion}]{gavrielov2024linear}
Gavrielov N, Ivrii A, Garion S (2024) Linear circuit synthesis using weighted
  steiner trees. \emph{Quantum Information and Computation} 24(9\&10):721--733.

\bibitem[{Glos et~al.(2022)Glos, Krawiec, \protect\BIBand{}
  Zimbor{\'a}s}]{HOBO}
Glos A, Krawiec A, Zimbor{\'a}s Z (2022) Space-efficient binary optimization
  for variational quantum computing. \emph{npj Quantum Information} 8(1):1--8.

\bibitem[{Herrman et~al.(2022)Herrman, Lotshaw, Ostrowski, Humble,
  \protect\BIBand{} Siopsis}]{herrman2022multi}
Herrman R, Lotshaw PC, Ostrowski J, Humble TS, Siopsis G (2022) Multi-angle
  quantum approximate optimization algorithm. \emph{Scientific Reports}
  12(1):6781.

\bibitem[{Herrman et~al.(2021)Herrman, Treffert, Ostrowski, Lotshaw, Humble,
  \protect\BIBand{} Siopsis}]{GOQAOA}
Herrman R, Treffert L, Ostrowski J, Lotshaw PC, Humble TS, Siopsis G (2021)
  {Globally optimizing QAOA circuit depth for constrained optimization
  problems}. \emph{Algorithms} 14(10):294.

\bibitem[{{IBM Quantum}(2019)}]{Qiskit2019}
{IBM Quantum} (2019) {Qiskit: An Open-source Framework for Quantum Computing}.

\bibitem[{Jung \protect\BIBand{} Choi(2025)}]{jung2025new}
Jung J, Choi IC (2025) A new multicommodity network flow model and branch and
  cut for optimal quantum boolean circuit synthesis. \emph{INFORMS Journal on
  Computing} 37(1):42--62.

\bibitem[{Koch et~al.(2026)Koch, Panahiyan, Mukherjee, Doetsch,
  \protect\BIBand{} Jaksch}]{koch2026resource}
Koch F, Panahiyan S, Mukherjee R, Doetsch J, Jaksch D (2026) Resource-efficient
  quantum optimization via higher-order encoding. \emph{EPJ Quantum Technology}
  .

\bibitem[{Laplace~Mermoud et~al.(2026)Laplace~Mermoud, Simonetto,
  \protect\BIBand{} Elloumi}]{laplace2025variational}
Laplace~Mermoud D, Simonetto A, Elloumi S (2026) Variational quantum algorithms
  for permutation-based combinatorial problems: {O}ptimal ansatz generation
  with applications to quadratic assignment problems and beyond.
  \emph{{Quantum}} 10:1998, ISSN 2521-327X,
  \urlprefix\url{http://dx.doi.org/10.22331/q-2026-02-09-1998}.

\bibitem[{Leven \protect\BIBand{} Galil(1983)}]{leven1983np}
Leven D, Galil Z (1983) Np completeness of finding the chromatic index of
  regular graphs. \emph{Journal of Algorithms} 4(1):35--44.

\bibitem[{Madden \protect\BIBand{} Simonetto(2022)}]{madden2022best}
Madden L, Simonetto A (2022) Best approximate quantum compiling problems.
  \emph{ACM Transactions on Quantum Computing} 3(2):1--29.

\bibitem[{Mandal et~al.(2020)Mandal, Roy, Upadhyay, \protect\BIBand{}
  Ushijima-Mwesigwa}]{mandal2020compressed}
Mandal A, Roy A, Upadhyay S, Ushijima-Mwesigwa H (2020) Compressed
  quadratization of higher order binary optimization problems.
  \emph{Proceedings of the 17th ACM International Conference on Computing
  Frontiers}, 126--131.

\bibitem[{Markov et~al.(2008)Markov, Patel, \protect\BIBand{}
  Hayes}]{markov2008optimal}
Markov K, Patel I, Hayes J (2008) Optimal synthesis of linear reversible
  circuits. \emph{Quantum Information and Computation} 8(3\&4):0282--0294.

\bibitem[{Maslov(2007)}]{maslov2007linear}
Maslov D (2007) Linear depth stabilizer and quantum fourier transformation
  circuits with no auxiliary qubits in finite-neighbor quantum architectures.
  \emph{Physical Review A—Atomic, Molecular, and Optical Physics}
  76(5):052310.

\bibitem[{Meijer-van~de Griend \protect\BIBand{} Li(2022)}]{de2022dynamic}
Meijer-van~de Griend A, Li SM (2022) Dynamic qubit routing with cnot circuit
  synthesis for quantum compilation. \emph{arXiv preprint arXiv:2205.00724} .

\bibitem[{Moro et~al.(2021)Moro, Paris, Restelli, \protect\BIBand{}
  Prati}]{moro2021quantum}
Moro L, Paris MG, Restelli M, Prati E (2021) Quantum compiling by deep
  reinforcement learning. \emph{Communications Physics} 4(1):178.

\bibitem[{Nagarajan et~al.(2021)Nagarajan, Lockwood, \protect\BIBand{}
  Coffrin}]{nagarajan2021quantumcircuitopt}
Nagarajan H, Lockwood O, Coffrin C (2021) Quantumcircuitopt: An open-source
  framework for provably optimal quantum circuit design. \emph{2021 IEEE/ACM
  Second International Workshop on Quantum Computing Software (QCS)}, 55--63
  (IEEE).

\bibitem[{Nagarajan \protect\BIBand{} Szabo(2025)}]{nagarajan2025provably}
Nagarajan H, Szabo Z (2025) Provably optimal quantum circuits with
  mixed-integer programming. \emph{arXiv preprint arXiv:2510.00649} .

\bibitem[{Nannicini et~al.(2022)Nannicini, Bishop, G{\"u}nl{\"u}k,
  \protect\BIBand{} Jurcevic}]{nannicini2022optimal}
Nannicini G, Bishop LS, G{\"u}nl{\"u}k O, Jurcevic P (2022) Optimal qubit
  assignment and routing via integer programming. \emph{ACM Transactions on
  Quantum Computing} 4(1):1--31.

\bibitem[{Nielsen \protect\BIBand{} Chuang(2002)}]{nielsenbook}
Nielsen MA, Chuang I (2002) \emph{Quantum computation and quantum information}
  (American Association of Physics Teachers).

\bibitem[{Rehfeldt et~al.(2023)Rehfeldt, Koch, \protect\BIBand{}
  Shinano}]{Rehfeldt2023}
Rehfeldt D, Koch T, Shinano Y (2023) Faster exact solution of sparse {MaxCut}
  and {QUBO} problems. \emph{Mathematical Programming Computation}
  15(3):445--470, ISSN 1867-2957,
  \urlprefix\url{http://dx.doi.org/10.1007/s12532-023-00236-6}.

\bibitem[{Romero et~al.(2025)Romero, Visuri, Cadavid, Simen, Solano,
  \protect\BIBand{} Hegade}]{romero2025bias}
Romero SV, Visuri AM, Cadavid AG, Simen A, Solano E, Hegade NN (2025)
  Bias-field digitized counterdiabatic quantum algorithm for higher-order
  binary optimization. \emph{Communications Physics} 8(1):348.

\bibitem[{Rotman(1999)}]{rotman1999introduction}
Rotman JJ (1999) \emph{An Introduction to the Theory of Groups}, volume 148
  (Springer Science \& Business Media).

\bibitem[{Seeley et~al.(2012)Seeley, Richard, \protect\BIBand{} Love}]{BK12}
Seeley JT, Richard MJ, Love PJ (2012) The {Bravyi-Kitaev} transformation for
  quantum computation of electronic structure. \emph{The Journal of chemical
  physics} 137(22):224109.

\bibitem[{Shaydulin et~al.(2024)Shaydulin, Li, Chakrabarti, DeCross, Herman,
  Kumar, Larson, Lykov, Minssen, Sun, Alexeev, Dreiling, Gaebler, Gatterman,
  Gerber, Gilmore, Gresh, Hewitt, Horst, Hu, Johansen, Matheny, Mengle, Mills,
  Moses, Neyenhuis, Siegfried, Yalovetzky, \protect\BIBand{}
  Pistoia}]{Shaydulin2024}
Shaydulin R, Li C, Chakrabarti S, DeCross M, Herman D, Kumar N, Larson J, Lykov
  D, Minssen P, Sun Y, Alexeev Y, Dreiling JM, Gaebler JP, Gatterman TM, Gerber
  JA, Gilmore K, Gresh D, Hewitt N, Horst CV, Hu S, Johansen J, Matheny M,
  Mengle T, Mills M, Moses SA, Neyenhuis B, Siegfried P, Yalovetzky R, Pistoia
  M (2024) Evidence of scaling advantage for the quantum approximate
  optimization algorithm on a classically intractable problem. \emph{Science
  Advances} 10(22), ISSN 2375-2548,
  \urlprefix\url{http://dx.doi.org/10.1126/sciadv.adm6761}.

\bibitem[{Shende et~al.(2004)Shende, Markov, \protect\BIBand{}
  Bullock}]{shende2004lowerbound}
Shende VV, Markov IL, Bullock SS (2004) {Minimal universal two-qubit
  controlled-NOT-based circuits}. \emph{Phys. Rev. A} 69:062321,
  \urlprefix\url{http://dx.doi.org/10.1103/PhysRevA.69.062321}.

\bibitem[{Sivarajah et~al.(2020)Sivarajah, Dilkes, Cowtan, Simmons, Edgington,
  \protect\BIBand{} Duncan}]{tket}
Sivarajah S, Dilkes S, Cowtan A, Simmons W, Edgington A, Duncan R (2020)
  {t$\vert$ket${\rangle}$:} a retargetable compiler for {NISQ} devices.
  \emph{Quantum Science and Technology} 6(1):014003.

\bibitem[{Thei{\ss}inger et~al.(2026)Thei{\ss}inger, Gerlach, Berghaus,
  \protect\BIBand{} Bauckhage}]{theissinger2026beyond}
Thei{\ss}inger L, Gerlach T, Berghaus D, Bauckhage C (2026) Beyond
  reinforcement learning: Fast and scalable quantum circuit synthesis.
  \emph{arXiv preprint arXiv:2602.15146} .

\bibitem[{Uotila et~al.(2025)Uotila, Ripatti, \protect\BIBand{}
  Zhao}]{uotila2025higher}
Uotila V, Ripatti J, Zhao B (2025) Higher-order portfolio optimization with
  quantum approximate optimization algorithm. \emph{2025 IEEE International
  Conference on Quantum Computing and Engineering (QCE)}, volume~1, 01--12
  (IEEE).

\bibitem[{Vandaele et~al.(2022)Vandaele, Martiel, \protect\BIBand{} Goubault~de
  Brugi{\`e}re}]{vandaele2022phase}
Vandaele V, Martiel S, Goubault~de Brugi{\`e}re T (2022) Phase polynomials
  synthesis algorithms for nisq architectures and beyond. \emph{Quantum Science
  \& Technology} 7(4):045027.

\bibitem[{Verch{\`e}re et~al.(2023)Verch{\`e}re, Elloumi, \protect\BIBand{}
  Simonetto}]{verchere2023optimizing}
Verch{\`e}re Z, Elloumi S, Simonetto A (2023) {Optimizing Variational Circuits
  for Higher-Order Binary Optimization}. \emph{{IEEE International Conference
  on Quantum Computing and Engineering (QCE)}} (Seattle (WA), United States),
  \urlprefix\url{https://hal.science/hal-04174561}, 6 pages plus references.
  Conference: IEEE QCE 2023.

\bibitem[{Webster et~al.(2026)Webster, Koutsioumpas, \protect\BIBand{}
  Browne}]{Webster2026heuristicoptimal}
Webster M, Koutsioumpas S, Browne DE (2026) Heuristic and {O}ptimal {S}ynthesis
  of {CNOT} and {C}lifford {C}ircuits. \emph{{Quantum}} 10.

\bibitem[{Welch et~al.(2014)Welch, Greenbaum, Mostame, \protect\BIBand{}
  Aspuru-Guzik}]{welch2014efficient}
Welch J, Greenbaum D, Mostame S, Aspuru-Guzik A (2014) Efficient quantum
  circuits for diagonal unitaries without ancillas. \emph{New Journal of
  Physics} 16(3):033040.

\bibitem[{Yan et~al.(2024)Yan, Wu, Chen, Pan, Lu, Zhou, Wang, Wang,
  \protect\BIBand{} Yan}]{yan2024quantum}
Yan G, Wu W, Chen Y, Pan K, Lu X, Zhou Z, Wang Y, Wang R, Yan J (2024) Quantum
  circuit synthesis and compilation optimization: Overview and prospects.
  \emph{arXiv preprint arXiv:2407.00736} .

\bibitem[{Younis et~al.(2021)Younis, Iancu, Lavrijsen, Davis, \protect\BIBand{}
  Smith}]{doecode58510}
Younis E, Iancu CC, Lavrijsen W, Davis M, Smith E (2021) Berkeley quantum
  synthesis toolkit (bqskit) v1. [Computer Software]
  \url{https://doi.org/10.11578/dc.20210603.2},
  \urlprefix\url{http://dx.doi.org/10.11578/dc.20210603.2}.

\bibitem[{Zhou et~al.(2020)Zhou, Wang, Choi, Pichler, \protect\BIBand{}
  Lukin}]{Zhou2018}
Zhou L, Wang ST, Choi S, Pichler H, Lukin MD (2020) Quantum approximate
  optimization algorithm: Performance, mechanism, and implementation on
  near-term devices. \emph{Physical Review X} 10(2):021067.

\bibitem[{Zylberman et~al.(2025)Zylberman, Nzongani, Simonetto,
  \protect\BIBand{} Debbasch}]{zylberman2025efficient}
Zylberman J, Nzongani U, Simonetto A, Debbasch F (2025) Efficient quantum
  circuits for non-unitary and unitary diagonal operators with
  space-time-accuracy trade-offs. \emph{ACM Transactions on Quantum Computing}
  6(2):1--43.

\end{thebibliography}

\clearpage
\newpage


{\huge{\bf Appendices}}

\appendix

\section{Proof of Theorem~\ref{thm:np-completeness:1}}\label{app.proof}

We need a few intermediary results before tackling the main proof. 

The core of the proof is a reduction from the 3-edge-coloring problem which consists, given a graph $\mathcal G = (V, E)$ in determining whether it is possible to color the edges of $\mathcal G$ with three colors in such a way that no two incident edges share the same color. This problem is \textsf{NP}-Complete even if $\mathcal G$ is 3-regular \citep{leven1983np}. The reduction will, given an instance $\mathcal G$ of 3-edge-coloring builds an instance $(P_{\mathcal G}, {\mathcal D}_{\mathcal G})$ of $(2, 9)$-\arpbac. We recall that we indicate with $P_{\mathcal G}(z)$ the polynomial in $z$ for a given graph $\mathcal G$, and with $H_{P, \mathcal G}(Z)$ its representation with Pauli-Z operators. 

We first do an intermediate reduction to transform $\mathcal G$ into a 3-regular graph $\mathcal G'$ that will be close to our topology. In order to do so, we expand each edge $e = (u,v) \in  E$ into a bipartite graph with 8 vertices $(u, e_1, e_2, \dots, e_6, v)$ as represented in Figure~\ref{fig:np-compelteness:1}. Note that any feasible coloring of that bipartite graph impose that the color of $(u, e_1)$ is the same as the color of $(v, e_6)$, thus a feasible coloring of $\mathcal G$ can be deduced from a 3-edge-coloring of $\mathcal G'$ where the color of $(u, e_1)$ is assigned to $e$, and conversely. Consequently, $\mathcal G$ is 3-edge-colorable if and only if $\mathcal  G'$ is. 

Given $\mathcal  G'$, we now build the polynomial $P_{\mathcal G}$ and the topology ${\mathcal D}_{\mathcal G}$. Note that the order of $\mathcal G'$ is $|V| + 6 |E| = 10|V|$ as $\mathcal G$ is 3-regular, we then set the number of variables of the polynomials $P_{\mathcal G}(z)$ to $n = 10|V|$. We define three sets $V_1$, $V_2$, and $V_3$, respectively, hereinafter called \emph{source, intermediate}, and \emph{sink vertices} of $\mathcal D_{\mathcal G}$. The source vertices $V_1$ contain each vertex originally in $\mathcal G$ with the nodes $e_4$ and $e_5$ for each edge $e \in E$. The intermediate vertices $V_2$ are the set $\{e_1, | e \in E\}$ and the sink vertices $V_3$ contain $e_2, e_3$ and $e_6$ for each $e \in E$. We now get the topology $\mathcal D_{\mathcal G}$ by directing every edge of $\mathcal G'$ from $V_1$ to $V_2$, from $V_2$ to $V_3$ and from $V_1$ to $V_3$. Observe first that the vertices in $\mathcal G'$ and $\mathcal D_{\mathcal G}$ are the same, secondly that $V_1, V_2$ and $V_3$ are independent sets, thus $\mathcal D_{\mathcal G}$ is a 3-regular acyclic digraph. 

    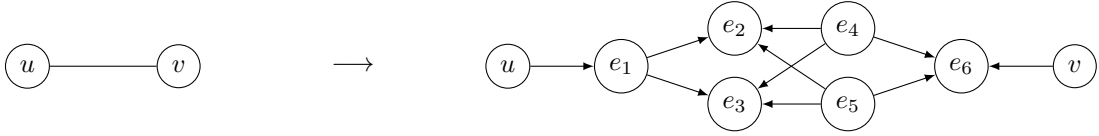
\begin{figure}[h!]
        \centering
        \begin{minipage}[c]{0.2\textwidth}
        \centering
        \begin{tikzpicture}
            \node[draw, circle] (U) at (0, 0) {$u$};
            \node[draw, circle] (V) at (2, 0) {$v$};
            \draw (U) -- (V);
        \end{tikzpicture}
        \end{minipage}
        \qquad
        \begin{minipage}[c]{0.1\textwidth}
        \centering
        $\longrightarrow$
        \end{minipage}
        \qquad
        \begin{minipage}[c]{0.5\textwidth}
        \centering
        \begin{tikzpicture}
            \node[draw, circle] (U) at (-0.5, 0) {$u$};
            \node[draw, circle] (W1) at (1, 0) {$e_1$};
            \node[draw, circle] (W2) at (2.5, 0.5) {$e_2$};
            \node[draw, circle] (W3) at (2.5, -0.5) {$e_3$};
            \node[draw, circle] (W4) at (4, 0.5) {$e_4$};
            \node[draw, circle] (W5) at (4, -0.5) {$e_5$};
            \node[draw, circle] (W6) at (5.5, 0) {$e_6$};
            \node[draw, circle] (V) at (7, 0) {$v$};
            \draw[>=latex, ->] (U) -- (W1);
            \draw[>=latex, ->] (W1) -- (W2);
            \draw[>=latex, ->] (W1) -- (W3);
            \draw[>=latex, ->] (W4) -- (W2);
            \draw[>=latex, ->] (W5) -- (W2);
            \draw[>=latex, ->] (W4) -- (W3);
            \draw[>=latex, ->] (W5) -- (W3);
            \draw[>=latex, ->] (W4) -- (W6);
            \draw[>=latex, ->] (W5) -- (W6);
            \draw[>=latex, ->] (V) -- (W6);
        \end{tikzpicture}
        \end{minipage}
        \caption{Reduction from $\mathcal G$ to $\mathcal G'$. That graph is then directed to obtain the topology $\mathcal D_{\mathcal G}$.}
        \label{fig:np-compelteness:1}
    \end{figure}

Finally, we associate to each vertex an index from 1 to $n$ by numbering arbitrarily the source vertices with \(\{1, \ldots, \lvert V_1 \rvert\}\), then the intermediate vertices with \(\{\lvert V_1 \rvert + 1, \ldots, \lvert V_1 \rvert + \lvert V_2 \rvert\}\) and finally the sink vertices with \(\{\lvert V_1 \rvert + \lvert V_2 \rvert + 1, \ldots, \lvert V_1 \rvert + \lvert V_2 \rvert + \lvert V_3 \rvert = n\}\). We get an ordered list $W = (w_1, w_2, \dots, w_n)$ of vertices of $\mathcal D_{\mathcal G}$. Note that this is a topological ordering of the vertices of $\mathcal D_{\mathcal G}$. Given $A$ the arcs of $\mathcal D_{\mathcal G}$ and $(w_i, w_j) \in A$, we have that $i < j$. The polynomial $P_{\mathcal G}(z)$ is set to $P_{\mathcal G}(z) = \sum_{(w_i, w_j) \in A} z_iz_j$. The degree of $P_{\mathcal G}(z)$ is, as expected, 2. 

The performed procedure is achieved in polynomial time. We now prove that there exists a feasible circuit for that instance $(P_{\mathcal G}, \mathcal D_{\mathcal G})$ if and only if $\mathcal G'$ is 3-colorable. 

\begin{lemma}
\label{lem:np-completeness:1}
    If $\mathcal G'$ is 3-edge colorable, there exists a circuit of depth 9 encoding $P_{\mathcal G}$ following the topology $\mathcal D_{\mathcal G}$.
\end{lemma}
\begin{proof}
Let $F_1, F_2, F_3$ be the partition of the sets of edges of $\mathcal G'$ colored with each of the three colors. By definition, each such set is a matching of $\mathcal G'$. The circuit contains, for each $k \in \{1,2,3\}$ three successive layers: 
\begin{itemize}
    \item (first layer) a \CNOT\ gate $\cnot{i}{j}$ for each $(w_i, w_j) \in F_k$
    \item (second layer) an \RZ\ gate $\rz{j}$ for each $(w_i, w_j) \in F_k$
    \item (third layer) a \CNOT\ gate $\cnot{i}{j}$ for each $(w_i, w_j) \in F_k$
\end{itemize}
Note that, first, after the third layer, every qubit represents only itself. Secondly, each monomial $z_iz_j$ is represented at the second, fifth or eighth layer depending if $(w_i, w_j)$ belongs to $F_1, F_2$ or $F_3$. The circuit contains 9 layers and satisfies the topology constraint, therefore it is feasible. 
\end{proof}

The rest of the proof is dedicated to show that, given a feasible circuit encoding $P_{\mathcal G}$, we can deduce a 3-edge-coloring of $\mathcal G'$. Due to the acyclic topology $\mathcal D_{\mathcal G}$, a feasible circuit for that instance must satisfy the following constraint: the circuit contains only downward \(\cnot{}{}\) gates. An immediate consequence is that, if $i < j$, then the former qubit cannot represent the latter. Therefore, the circuit cannot contain any \textsc{swap} circuit and any monomial $z_iz_j$ must be represented on the qubit $j$: there are three \RZ\ gates on each qubit associated with a sink vertex and one for the intermediate vertices. In the following lemmas we demonstrate that every \RZ\ gate is either at layer 2, 5 or 8. This partition of the $\rz{}$ gates allows us to build a 3-edge-coloring of $\mathcal G'$. For convenience, hereinafter we call $\rz{i,j}$ the $\rz{}$ gate that represents the phase associated to the monomial $z_iz_j$. We also say that a \CNOT\ gate sends a variable if a monomial $z_i$ is routed through it to form another monomial $z_i z_j$. In the same way a \CNOT\ gate can send multiple variables, if $z_i z_j \cdots$ is routed through it. 

\begin{lemma}
    \label{lem:np-completeness:2}
    Assuming that there exists a feasible 9-depth circuit for $(P_{\mathcal G}, {\mathcal D}_{\mathcal G})$, if the circuit contains a \CNOT\ gate sending two variables or more, there is another feasible 9-depth circuit containing only \CNOT\ gates sending one variable. 
\end{lemma}
\begin{proof}
    Let $\cnot{i}{j}$ be a \CNOT\ gate sending a variable $z_k$ ($k \neq i$), then there exists a directed path in $\mathcal D_{\mathcal G}$ connecting $w_k$ to $w_j$ through the vertex $w_i$. By construction, the only possible cases are the paths of size $3$: $(u, e_1, e_2)$ and $(u, e_1, e_3)$. In the two cases, the gate $\cnot{i}{j}$ is the only way to send $z_k$ to $j$. As $j$ does not represent $z_k$ after layer 9, there are an even number of such gates. Also, $P_{\mathcal G}$ does not contain the monomial $z_kz_j$, so while the variable $z_k$ is represented on the qubit $j$, no \RZ\ gate may be placed on that qubit. Thus the circuit where those gates are removed is still feasible. 
\end{proof}

\begin{lemma}
    \label{lem:np-completeness:3}
Let $w_j$ be a sink vertex. Assuming that there exists a feasible 9-depth circuit for $(P_{\mathcal G}, {\mathcal D}_{\mathcal G})$ where every \(\cnot{}{}\) gate sends only one variable, there are three \RZ\ gates $\rz{i_1, j}$, $\rz{i_2, j}$ and $\rz{i_3, j}$ respectively placed at layers 2, 5 and 8 on qubit $j$ where $w_{i_1}$, $w_{i_2}$, $w_{i_3}$ are the three neighbors of $w_j$ in $\mathcal G'$. In addition, two gates $\cnot{i_1}{j}$ are placed at layers 1 and either 3 or 4, two gates $\cnot{i_2}{j}$ are placed at layers either 3 or 4 and either 6 or 7, and, two gates $\cnot{i_3}{j}$ are placed at layers either 6 or 7 and 9. 
\end{lemma}
\begin{proof}

By construction, the three neighbors of $w_j$ are in $V_1 \cup V_2$. Thus $i_1, i_2, i_3 \leq j$ and the three gates $\rz{i_1, j}$, $\rz{i_2, j}$ and $\rz{i_3, j}$ are placed on the qubit $j$. As $j$ represents only $z_j$ before the first layer, at least one \(\cnot{}{}\) gate sending $i_1$ on $j$ is placed before $\rz{i_1, j}$. As the circuit contains no \textsc{swap} gate, this \CNOT\ gate is $\cnot{i_1}{j}$. Then one or more \CNOT\ gates $\cnot{i_1}{j}$ and $\cnot{i_2}{j}$ on $j$ are placed between $\rz{i_1, j}$ and $\rz{i_2, j}$. Symmetrically, there are at least two \CNOT\ gates $\cnot{i_2}{j}$ and $\cnot{i_3}{j}$ between $\rz{i_2, j}$ and $\rz{i_3, j}$ and one after $\rz{i_3, j}$. As the depth of the circuit is 9, we deduce that no additional \CNOT\ gate may be placed: the layers 2, 5 and 8 contain the three \RZ\ gates and the six other layers contain a \CNOT\ gate. 
\end{proof}

\begin{lemma}
    \label{lem:np-completeness:4}
Let $w_j$ be an intermediate vertex. Assuming that there exists a feasible 9-depth circuit for $(P_{\mathcal G}, {\mathcal D}_{\mathcal G})$ where every \CNOT\ gate sends only one variable, there is one \RZ\ gate $\rz{i_1, j}$ placed at layers 2, 5 or 8 on qubit $j$ where $w_{i_1}$ is the neighbors of $w_j$ in $V_1$. In addition, two gates $\cnot{i_1}{j}$ are placed at layers 1 and 3 or 4 if $\rz{i_1,j}$ is at layer 2, at layers 3 or 4 and 6 or 7 if $\rz{i_1,j}$ is at layer 5, and at layers 6 or 7 and 9 if $\rz{i_1,j}$ is at layer 8.
\end{lemma}
\begin{proof}
Let $w_{i_2} \in V_3$, $w_{i_3} \in V_3$ be the two neighbors of $w_j$ in $V_3$. The gate $\rz{i_1,j}$ is placed on the qubit $j$ as $i_1 < j$. That gate is surrounded by two CNOT gates $\cnot{i_1}{j}$ and is then not at layers 1 and 9. Let assume $\rz{i_1,j}$ is placed at layer 3. According to Lemma~\ref{lem:np-completeness:3}, as $w_{i_2}$ is a sink vertex, there are two \CNOT\ gates $\cnot{j}{i_2}$ in the circuit sending only the variable $z_j$ on $i_2$. As $i_1$ is represented by $j$ between the two gates $\cnot{i_1}{j}$, the gates $\cnot{j}{i_2}$ cannot be placed between them, thus not at layer 2, 3 or 4. However, by Lemma~\ref{lem:np-completeness:3}, if the gate $\rz{j, i_2}$ is placed at layer 2 or 5 of the qubit $i_2$, there is one gate $\cnot{j}{i_2}$ at layer 3 or 4, which is a contradiction. Thus, $\rz{j, i_2}$ is at layer 8 and there is a gate $\cnot{j}{i_2}$ at layer 9. The same demonstration can be made for $i_3$. Hence there are two \CNOT\ gates at layer 9 on qubit $j$ and this is a contradiction. Therefore, $\rz{i_1,j}$ is not at layer 3. Symmetrically, is it not at layer 4, 6 or 7. Thus, it is either at layer 2, 5 or 8. A similar argument can be used to prove that the two gates $\cnot{i_1}{j}$ are placed at the layer preceding and following the gate $\rz{i_1,j}$, otherwise there is no room for the two gates $\cnot{j}{i_2}$ and the two gates $\cnot{j}{i_3}$. 
\end{proof}

\begin{lemma}
    \label{lem:np-completeness:5}
Assuming that there exists a feasible 9-depth circuit for $(P_{\mathcal G}, {\mathcal D}_{\mathcal G})$ where every \CNOT\ gate sends only one variable, then $\mathcal G'$ is 3-edge-colorable.
\end{lemma}
\begin{proof}
As no source vertex has an entering arc, every \RZ\ gate is on a qubit associated with an intermediate vertex or a sink vertex. By Lemmas~\ref{lem:np-completeness:3} and \ref{lem:np-completeness:4} all those gates are placed on layers 2, 5 or 8. We may now define a partition $F_1, F_2$ and $F_3$ of the edges of $\mathcal G'$ with those sets being respectively all the edges $(w_i, w_j)$ such that $\rz{i, j}$ is at layer 2, 5 or 8. We finally prove that this partition defines a 3-edge-coloring of $\mathcal G'$ by showing that each set is a matching. By definition, two arcs $(w_{i_1}, w_{j_1})$ and $(w_{i_2}, w_{j_2})$ in $F_k$, for any $k \in \{1, 2, 3\}$ satisfy $j_1 \neq j_2$ as the two gates $\rz{i_1, j_1}$ and $\rz{i_2, j_2}$ cannot be on the same qubit and at the same layer. If the first set $F_1$ is not a matching then it contains two arcs $(i, j_1)$ and $(i, j_2)$. The two gates $\rz{i, j_1}$ and $\rz{i, j_2}$ are placed at layer 2. However, by Lemmas~\ref{lem:np-completeness:3} and \ref{lem:np-completeness:4}, this implies that there are two \CNOT\ gates $\cnot{i}{j_1}$ and $\cnot{i}{j_2}$ at layer 1 which is a contradiction. Consequently $F_1$ is a matching. For the same reason, $F_3$ also a matching. Assuming that $F_2$ is not a matching, there are two gates $\rz{i, j_1}$ and $\rz{i, j_2}$ placed at layer 5. Hence, by Lemmas~\ref{lem:np-completeness:3} and \ref{lem:np-completeness:4}, there are four \CNOT\ gates, $\cnot{i}{j_1}$ and $\cnot{i}{j_2}$, each placed twice, at layers 3, 4, 6 and 7. However, let $w_{j_3}$ be the third neighbor of $w_i$. Again, by Lemmas~\ref{lem:np-completeness:3} and \ref{lem:np-completeness:4}, if $j_3 < i$, there must be a gate $\cnot{j_3}{i}$ at layer 3 or 4 or 6 or 7, depending if the gate $\rz{j_3, i}$ is at layer 2, 5 or 8. Similarly if $i < j_3$. There is then a contradiction with the four first \CNOT\ gates: $F_2$ is a matching. Consequently, $\mathcal G'$ is 3-edge colorable. 
\end{proof}

We are now ready to prove Theorem~\ref{thm:np-completeness:1}

\begin{proof}
    By Lemmas~\ref{lem:np-completeness:1}, \ref{lem:np-completeness:2} and \ref{lem:np-completeness:5}, the described procedure is a polynomial reduction from the 3-edge-coloring problem to (2, 9)-\arpbac, thus this problem is \textsf{NP}-Complete.
\end{proof}

\section{Proof of Proposition~\ref{prop:1}}\label{ap:prop:1}

To represent a polynomial of degree $n$, we need all the unique combinations of indices up to degree $n$, of which there are $2^{n}-1$. With our construction, $\Te_{n}$ has $2^{n-1}-1$ rotations coming from $\Te_{n-1}$ plus one considering $\alpha_n$. To this, we add all the rotations on the second $2^{n-1}$ layers, i.e., all the rotations on qubit $n-1$, as $2^{n-2}-1$ and all the ones on qubit $n$, as $2^{n-2}$. Adding up we obtain a total of unique rotation gates of $2^{n}-1$, as required. 

As for the layers: $\Te_{n}$ has firstly $2^{n-1}$ layers by construction, at which we add all the \CNOTs\ plus \RZ\ gates of these first $2^{n-1}$ layers a second time, adding another $2^{n-1}$ layers, totalling at $2^{n}$. Note that the circuit acting on qubit $n-1$ on the second half can happen in parallel with the one acting on qubit $n$, so it is not counted here. \Halmos

\section{Simple examples of the application of \VT}\label{app.examples}

\begin{example}
Let us use our algorithm on the following polynomial: 
\[
P(z) = z_1 z_2 + z_1 z_2 z_3 + z_1 z_2 z_3 z_4 + z_2 z_4. 
\]
First, we apply the template to it, by putting the coefficients of the missing monomials to \(0\). We obtain the following circuit, depicted in Figure~\ref{fig: ex-template-simplification}. 

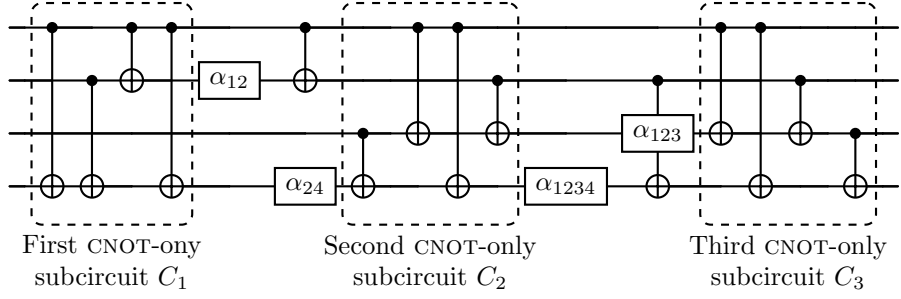
\begin{figure}[h]
    \centering
    \begin{quantikz}[row sep = {0.7cm,between origins}, column sep = 0.2cm]
    & \qw & \ctrl{3}\gategroup[4, steps = 4, style = {dashed, rounded corners, inner xsep = 0pt}, label style = {label position = below, anchor = north, yshift = -0.2cm}]{\begin{minipage}{3.8cm} \begin{center} First \CNOT -ony subcircuit \(C_1\) \end{center} \end{minipage}} & \qw & \ctrl{1} & \ctrl{3} & \qw & \ctrl{1} & \qw\gategroup[4, steps = 5, style = {dashed, rounded corners, inner xsep = 0pt}, label style = {label position = below, anchor = north, yshift = -0.2cm}]{\begin{minipage}{3.8cm} \begin{center} Second \CNOT -only subcircuit \(C_2\) \end{center} \end{minipage}} & \qw & \ctrl{2} & \ctrl{3} & \qw & \qw & \qw & \ctrl{2}\gategroup[4, steps = 5, style = {dashed, rounded corners, inner xsep = 0pt}, label style = {label position = below, anchor = north, yshift = -0.2cm}]{\begin{minipage}{3.8cm} \begin{center} Third \CNOT -only subcircuit \(C_3\) \end{center} \end{minipage}} & \ctrl{3} & \qw & \qw & \qw & \qw & \qw \\
    & \qw      & \qw & \ctrl{2} & \targ{} & \qw & \gate{\alpha_{12}} & \targ{} & \qw & \qw & \qw & \qw & \ctrl{1} & \qw & \ctrl{2} & \qw & \qw & \ctrl{1} & \qw & \qw & \qw & \qw \\
    & \qw      & \qw & \qw & \qw & \qw & \qw & \qw & \ctrl{1} & \qw & \targ{} & \qw & \targ{} & \qw & \gate{\alpha_{123}} & \targ{} & \qw & \targ{} & \qw &\ctrl{1} & \qw & \qw \\
    & \qw & \targ{} & \targ{} & \qw & \targ{} & \qw & \gate{\alpha_{24}} & \targ{} & \qw & \qw & \targ{} & \qw & \gate{\alpha_{1234}} & \targ{} & \qw & \targ{} & \qw & \qw & \targ{} & \qw & \qw
    \end{quantikz}
    \caption{Template circuit for \(P\). The subcircuits \(C_1\), \(C_2\) and \(C_3\) only contain \CNOT\ gates.}
    \label{fig: ex-template-simplification}
\end{figure}

In this circuit, because some rotation gates are missing, we observe that there are some sequences of \CNOT\ that may contain some redundancy, meaning that we could rewrite this part of the circuit with possibly less gates. For the subcircuit \(C_1\), the corresponding matrix is
\[
\begingroup
\renewcommand*{\arraystretch}{0.7}
A_1 = \begin{pmatrix} 1 & 0 & 0 & 0 \\ 0 & 1 & 0 & 0 \\ 0 & 0 & 1 & 0 \\ 1 & 0 & 0 & 1 \end{pmatrix} \begin{pmatrix} 1 & 0 & 0 & 0 \\ 1 & 1 & 0 & 0 \\ 0 & 0 & 1 & 0 \\ 0 & 0 & 0 & 1 \end{pmatrix} \begin{pmatrix} 1 & 0 & 0 & 0 \\ 0 & 1 & 0 & 0 \\ 0 & 0 & 1 & 0 \\ 0 & 1 & 0 & 1 \end{pmatrix} \begin{pmatrix} 1 & 0 & 0 & 0 \\ 0 & 1 & 0 & 0 \\ 0 & 0 & 1 & 0 \\ 1 & 0 & 0 & 1 \end{pmatrix} = \begin{pmatrix} 1 & 0 & 0 & 0 \\ 1 & 1 & 0 & 0 \\ 0 & 0 & 1 & 0 \\ 0 & 1 & 0 & 1 \end{pmatrix}.
\endgroup
\]
By looking at the numbers of non-zero off-diagonal term of \(A_1\), we know that we can rewrite the subcircuit \(C_1\), that originally contained \(4\) \CNOT\ gates, with only \(2\) \CNOT\ gates, namely \(\cnot{2}{4}\) followed by \(\cnot{1}{2}\). The newly obtained circuit is depicted in Figure~\ref{fig: ex-template-simplification-final}. \Halmos

\begin{figure}[h]
    \centering
    \begin{quantikz}[row sep = {0.7cm,between origins}, column sep = 0.2cm]
    & \qw & \qw & \ctrl{1} & \qw & \ctrl{1} & \qw & \qw & \ctrl{2} & \ctrl{3} & \qw & \qw & \ctrl{2} & \ctrl{3} & \qw & \qw & \qw & \qw & \qw \\
    & \qw & \ctrl{2} & \targ{} & \gate{\alpha_{12}} & \targ{} & \qw & \qw & \qw & \qw & \ctrl{1} & \ctrl{2} & \qw & \qw & \ctrl{1} & \qw & \qw & \qw & \qw \\
    & \qw & \qw & \qw & \qw & \qw & \ctrl{1} & \qw & \targ{} & \qw & \targ{} & \gate{\alpha_{123}} & \targ{} & \qw & \targ{} & \qw &\ctrl{1} & \qw & \qw \\
    & \qw & \targ{} & \qw & \qw & \gate{\alpha_{24}} & \targ{} & \qw & \qw & \targ{} & \gate{\alpha_{1234}} & \targ{} & \qw & \targ{} & \qw & \qw & \targ{} & \qw & \qw
    \end{quantikz}
    \caption{Final circuit.}
    \label{fig: ex-template-simplification-final}
\end{figure}
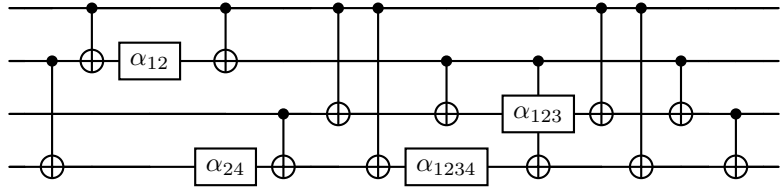
\end{example}

To conclude the presentation of the template-pruning method, we describe two examples that exhibits some non-trivial behavior. 

\begin{example}
    Consider the following circuit, depicted in Figure~\ref{fig:ex-swap} (left). 
\begin{figure}[H]
\centering
\begin{quantikz}[row sep = {0.7cm,between origins}, column sep = 0.2cm]
    \qw & \ctrl{1} & \qw & \ctrl{2} & \qw \\
    \qw & \targ{}  & \ctrl{1} & \qw & \qw \\
    \qw & \qw & \targ{} & \targ{} & \qw 
\end{quantikz} $\qquad \Rightarrow \qquad$ \begin{quantikz}[row sep = {0.7cm,between origins}, column sep = 0.2cm]
    \qw & \qw & \ctrl{1} & \qw \\
    \qw & \ctrl{1} & \targ{} & \qw \\
    \qw & \targ{} & \qw & \qw
\end{quantikz}
\caption{Initial (left) and final (right) circuits.}
\label{fig:ex-swap}
\end{figure}
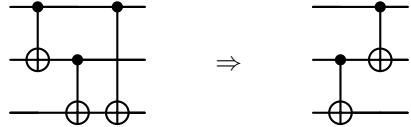
To translate this circuit into a matrix, we read it from right to left. The rightmost \CNOT\ gate is controlled by the first qubit, targeting the third one. The matrix representing it is 
\[
\begingroup
\renewcommand*{\arraystretch}{0.7}
{\bf T}^{(31)} = \begin{pmatrix}
1 & 0 & 0 \\ 0 & 1 & 0 \\ 1 & 0 & 1 
\end{pmatrix}.
\endgroup
\]
The second rightmost \CNOT\ gate is controlled by the second qubit, targeting the third one. So far, the matrix representing the part of the circuit we read is 
\[
\begingroup
\renewcommand*{\arraystretch}{0.7}
A = {\bf T}^{(31)} {\bf T}^{(32)} = \begin{pmatrix} 1 & 0 & 0 \\ 0 & 1 & 0 \\ 1 & 0 & 1 \end{pmatrix} \begin{pmatrix} 1 & 0 & 0 \\ 0 & 1 & 0 \\ 0 & 1 & 1 \end{pmatrix} = \begin{pmatrix} 1 & 0 & 0 \\ 0 & 1 & 0 \\ 1 & 1 & 1 \end{pmatrix}. 
\endgroup
\]
The next and last \CNOT\ gate is controlled by the first qubit, targeting the second one, so the whole circuit is represented by
\[
\begingroup
\renewcommand*{\arraystretch}{0.7}
A {\bf T}^{(21)} = \begin{pmatrix} 1 & 0 & 0 \\ 0 & 1 & 0 \\ 1 & 1 & 1 \end{pmatrix} \begin{pmatrix} 1 & 0 & 0 \\ 1 & 1 & 0 \\ 0 & 0 & 1 \end{pmatrix} = \begin{pmatrix} 1 & 0 & 0 \\ 1 & 1 & 0 \\ 0 & 1 & 1 \end{pmatrix}.
\endgroup
\]
We can already read from the matrix that the returned circuit will be of size \(2\), because there are two non-zero entries below the diagonal. The circuit will be composed of the \CNOT\ gates \(\cnot{1}{2}\) and \(\cnot{2}{3}\). We now determined in which order they should be added to the circuit. We start reading the matrix from the bottom right, we read the whole row from right to left, and we do the same on the row above, etc. So here the first non-zero entry we read this way is the one on the third row and second column, so \(\cnot{2}{3}\) comes first. Then, the final circuit is represented in Figure~\ref{fig:ex-swap} (right).

This example is interesting because it shows how a circuit can be rewritten by dropping a \CNOT\ gate, but it also swapped the two other \CNOT\ gates. So the pruning does not simply detect redundant gates one by one, it takes the whole block into account. \Halmos 
\end{example}

\begin{example}
    Consider the following circuit, depicted in Figure~\ref{fig:ex-longer} (left). 
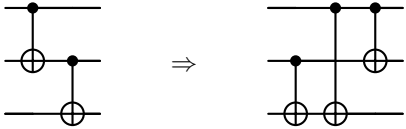
\begin{figure}[h]
\centering
\begin{quantikz}[row sep = {0.7cm,between origins}, column sep = 0.2cm]
    \qw & \ctrl{1} & \qw & \qw \\
    \qw & \targ{} & \ctrl{1} & \qw \\
    \qw & \qw & \targ{} & \qw
\end{quantikz} $\qquad \Rightarrow \qquad$ \begin{quantikz}[row sep = {0.7cm,between origins}, column sep = 0.2cm]
    \qw & \qw      & \ctrl{2} & \ctrl{1} & \qw \\
    \qw & \ctrl{1} & \qw      & \targ{}  & \qw \\
    \qw & \targ{}  & \targ{}  & \qw      & \qw
\end{quantikz}
\caption{Initial (left) and final (right) circuits.}
\label{fig:ex-longer}
\end{figure}

    Its matrix representation is given by 
    \[
    \begingroup
\renewcommand*{\arraystretch}{0.7}
    {\bf T}^{(32)} {\bf T}^{(21)} = \begin{pmatrix} 1 & 0 & 0 \\ 0 & 1 & 0 \\ 0 & 1 & 1 \end{pmatrix} \begin{pmatrix} 1 & 0 & 0 \\ 1 & 1 & 0 \\ 0 & 0 & 1 \end{pmatrix} = \begin{pmatrix} 1 & 0 & 0 \\ 1 & 1 & 0 \\ 1 & 1 & 1 \end{pmatrix}. 
    \endgroup
    \]
    By applying our algorithm, we obtain a circuit of size \(3\), because of the three non-zero entries below the diagonal. The final circuit is depicted in Figure~\ref{fig:ex-longer} (right).
    
This example shows that the transformation of a \CNOT -only circuit into its canonical form can make it longer. However, we show in the results that in general the returned circuits are much shorter than the original ones. \Halmos
\end{example}

\section{Proof of Proposition~\ref{prop.2}}\label{ap:prop:2}

There are $m$ non-zero monomials, so \VT\ outputs a circuit with $m$ \RZ\ gates. Note that, given the nature of the template (and owning to the Gray code), besides the first layer, each of these \RZ\ gates are not in parallel. Between each consecutive \RZ\ gates, we need at most $2 n$ layers of \CNOT\ gates, due to the transvection construction and~\citep[Theorem~4]{laplace2025variational}. So we need $2 m n + m$ layers (the $m$ due to the \RZ\ gates). \Halmos

It is interesting to now that even though the degree of the polynomial is $d$, there may be multiple \CNOT\ gates that belong to different qubits between consecutive \RZ\ gates. So we cannot reduce them in $2 d$ layers.

\section{Tables of experimental results}\label{app.results}

All the depth below are reported considering also an initial Hadamard tower and a final \texttt{rx} tower for QAOA implementation. Note that we have also tried a different simplification (instead of our pruning), called greedy Gaussian elimination (GGE, adapted from the more general~\cite{de2021gaussian}), with worse depth results than our pruning, and we have chosen not to include it in the tables.  

\subsection{Random instances}

\bigskip
\begin{center}
\centering
\adjustbox{max width=1.\textwidth}{\renewcommand{\arraystretch}{1.}
\begin{tabular}{ccccccccccccccccccccccccc}
\toprule
& & & & \multicolumn{4}{c}{Baselines} & \multicolumn{4}{c}{Ansatze} & \multicolumn{6}{c}{Binary linear programs} \\
\cmidrule(lr){5-8} \cmidrule(lr){9-12} \cmidrule(lr){13-18}
\multicolumn{3}{c}{Instances} & \hspace{0.2cm} & \multicolumn{2}{c}{Qiskit} & \multicolumn{2}{c}{Tket} & \multicolumn{2}{c}{\VT} & \multicolumn{2}{c}{Q-\VT} & \multicolumn{2}{c}{\CDP} & \multicolumn{2}{c}{\TB} & \multicolumn{2}{c}{\TBFM} \\
\cmidrule(lr){1-3} \cmidrule(lr){5-6} \cmidrule(lr){7-8} \cmidrule(lr){9-10} \cmidrule(lr){11-12} \cmidrule(lr){13-14} \cmidrule(lr){15-16} \cmidrule(lr){17-18}
$n$ & $d$ & $|M|$ &  & Depth & CPU [s] & Depth & CPU [s] & Depth & CPU [s] & Depth & CPU [s] & Depth & CPU [s] & Depth & CPU [s] & Depth & CPU [s] \\
\midrule
3 & 2 & 2 &  & 8 & 0.01 & 8 & 0.07 & 11 & \textbf{0.00} & 9 & 0.01 & \textbf{6} & 0.07 & \textbf{6} & 0.03 & \textbf{6} & 0.03 \\
3 & 2 & 3 &  & 11 & 0.02 & 11 & 0.09 & 10 & \textbf{0.00} & 10 & 0.01 & \textbf{8} & 0.15 & \textbf{8} & 0.06 & \textbf{8} & 0.04 \\
\midrule
5 & 2 & 2 &  & \textbf{5} & 0.01 & \textbf{5} & 0.03 & 10 & \textbf{0.01} & 10 & 0.01 & \textbf{5} & 0.26 & \textbf{5} & 0.14 & \textbf{5} & 0.12 \\
5 & 2 & 5 &  & 14 & 0.01 & 14 & 0.11 & 19 & \textbf{0.01} & 19 & 0.02 & \textbf{9} & 35.56 & \textbf{9} & 0.55 & 10 & 0.30 \\
5 & 2 & 6 &  & 14 & 0.02 & 14 & 0.22 & 19 & \textbf{0.01} & 19 & 0.01 & \textbf{9} & 31.01 & 10 & 1.66 & 13 & 0.38 \\
5 & 2 & 8 &  & 20 & 0.01 & 20 & 0.24 & 22 & \textbf{0.01} & 20 & 0.02 & \textbf{11} & 600.00 & \textbf{11} & 1.33 & 15 & 0.47 \\
5 & 2 & 10 &  & 23 & 0.01 & 23 & 0.31 & 25 & \textbf{0.01} & 23 & 0.02 & 13 & 600.00 & \textbf{12} & 6.41 & 17 & 0.50 \\
5 & 4 & 4 &  & 15 & 0.01 & 13 & 0.16 & 13 & \textbf{0.01} & 13 & 0.01 & \textbf{8} & 3.27 & 9 & 0.38 & 9 & 0.19 \\
5 & 4 & 7 &  & 21 & 0.01 & 17 & 0.21 & 17 & \textbf{0.01} & 16 & 0.02 & \textbf{9} & 70.01 & 10 & 1.16 & 12 & 0.29 \\
5 & 4 & 15 &  & 46 & 0.02 & 41 & 0.63 & 28 & \textbf{0.01} & 27 & 0.02 & 21 & 600.00 & \textbf{16} & 4.83 & 18 & 0.58 \\
5 & 4 & 20 &  & 53 & 0.02 & 47 & 0.70 & 31 & \textbf{0.01} & 31 & 0.03 & 27 & 600.00 & \textbf{20} & 27.80 & 22 & 0.77 \\
5 & 5 & 3 &  & 21 & 0.01 & 13 & 0.21 & 15 & \textbf{0.01} & 15 & 0.02 & -- & -- & -- & -- & \textbf{10} & 0.22 \\
5 & 5 & 12 &  & 30 & 0.02 & 29 & 0.43 & 24 & \textbf{0.01} & 22 & 0.02 & -- & -- & -- & -- & \textbf{18} & 0.52 \\
5 & 5 & 16 &  & 53 & 0.02 & 50 & 0.83 & 25 & \textbf{0.01} & 25 & 0.02 & -- & -- & -- & -- & \textbf{21} & 0.64 \\
5 & 5 & 23 &  & 67 & 0.02 & 64 & 1.06 & 33 & \textbf{0.01} & 31 & 0.02 & -- & -- & -- & -- & \textbf{29} & 0.91 \\
5 & 5 & 26 &  & 72 & 0.02 & 67 & 1.16 & 33 & \textbf{0.01} & 33 & 0.02 & -- & -- & -- & -- & \textbf{32} & 0.93 \\
\midrule
8 & 2 & 2 &  & 8 & 0.01 & 8 & 0.08 & 10 & \textbf{0.01} & 10 & 0.02 & \textbf{6} & 2.12 & \textbf{6} & 0.34 & \textbf{6} & 0.56 \\
8 & 2 & 9 &  & 20 & 0.02 & 20 & 0.32 & 29 & \textbf{0.01} & 22 & 0.03 & 12 & 600.00 & \textbf{11} & 26.88 & 18 & 1.99 \\
8 & 2 & 15 &  & 26 & \textbf{0.02} & 26 & 0.51 & 43 & 0.02 & 35 & 0.03 & -- & 600.00 & \textbf{13} & 112.91 & 26 & 5.78 \\
8 & 2 & 22 &  & 38 & \textbf{0.02} & 38 & 0.72 & 54 & 0.03 & 42 & 0.04 & -- & 600.00 & \textbf{18} & 600.00 & 35 & 5.28 \\
8 & 2 & 28 &  & 41 & \textbf{0.02} & 41 & 0.96 & 65 & 0.04 & 52 & 0.05 & -- & 600.00 & \textbf{20} & 600.00 & 43 & 7.84 \\
8 & 4 & 9 &  & 46 & \textbf{0.02} & 34 & 0.73 & 39 & 0.02 & 32 & 0.03 & -- & 600.00 & \textbf{15} & 45.22 & 17 & 3.24 \\
8 & 4 & 48 &  & 234 & \textbf{0.03} & 167 & 3.58 & 116 & 0.06 & 106 & 0.07 & -- & 600.00 & 111 & 600.00 & \textbf{69} & 21.39 \\
8 & 4 & 54 &  & 227 & \textbf{0.03} & 168 & 4.03 & 120 & 0.08 & 105 & 0.07 & -- & 600.00 & 123 & 600.00 & \textbf{71} & 600.00 \\
8 & 4 & 85 &  & 399 & \textbf{0.03} & 243 & 5.76 & 153 & 0.09 & 148 & 0.09 & -- & 600.00 & 172 & 600.00 & \textbf{102} & 43.02 \\
8 & 4 & 106 &  & 472 & \textbf{0.04} & 269 & 6.88 & 176 & 0.17 & 176 & 0.16 & -- & 600.00 & -- & 600.00 & \textbf{123} & 49.83 \\
8 & 5 & 30 &  & 156 & \textbf{0.02} & 122 & 2.84 & 88 & 0.04 & 83 & 0.05 & -- & -- & -- & -- & \textbf{47} & 65.00 \\
8 & 5 & 62 &  & 322 & \textbf{0.03} & 211 & 4.38 & 136 & 0.06 & 131 & 0.08 & -- & -- & -- & -- & \textbf{92} & 48.44 \\
8 & 5 & 88 &  & 413 & \textbf{0.03} & 302 & 6.28 & 164 & 0.08 & 154 & 0.12 & -- & -- & -- & -- & \textbf{120} & 48.04 \\
8 & 5 & 125 &  & 586 & \textbf{0.04} & 391 & 8.24 & 201 & 0.11 & 195 & 0.13 & -- & -- & -- & -- & \textbf{158} & 58.30 \\
8 & 5 & 148 &  & 712 & \textbf{0.05} & 448 & 9.56 & 212 & 0.12 & 210 & 0.14 & -- & -- & -- & -- & \textbf{172} & 66.71 \\
8 & 6 & 45 &  & 227 & \textbf{0.02} & 176 & 4.25 & 109 & 0.05 & 97 & 0.08 & -- & -- & -- & -- & \textbf{66} & 34.52 \\
8 & 6 & 79 &  & 333 & \textbf{0.03} & 274 & 6.12 & 150 & 0.07 & 132 & 0.10 & -- & -- & -- & -- & \textbf{105} & 36.57 \\
8 & 6 & 111 &  & 575 & \textbf{0.04} & 402 & 9.29 & 195 & 0.09 & 184 & 0.12 & -- & -- & -- & -- & \textbf{152} & 55.54 \\
8 & 6 & 154 &  & 796 & \textbf{0.06} & 507 & 11.45 & 235 & 0.15 & 230 & 0.21 & -- & -- & -- & -- & \textbf{191} & 72.55 \\
8 & 6 & 190 &  & 951 & \textbf{0.06} & 589 & 13.30 & 252 & 0.14 & 252 & 0.20 & -- & -- & -- & -- & \textbf{214} & 86.65 \\
8 & 7 & 52 &  & 250 & \textbf{0.04} & 201 & 4.81 & 121 & 0.06 & 110 & 0.11 & -- & -- & -- & -- & \textbf{76} & 600.00 \\
8 & 7 & 83 &  & 456 & \textbf{0.05} & 304 & 7.15 & 170 & 0.08 & 155 & 0.13 & -- & -- & -- & -- & \textbf{118} & 41.69 \\
8 & 7 & 131 &  & 690 & \textbf{0.08} & 469 & 11.22 & 213 & 0.11 & 204 & 0.19 & -- & -- & -- & -- & \textbf{170} & 102.25 \\
8 & 7 & 174 &  & 896 & 0.16 & 589 & 13.40 & 254 & \textbf{0.13} & 251 & 0.20 & -- & -- & -- & -- & \textbf{212} & 79.66 \\
8 & 7 & 226 &  & 1216 & \textbf{0.10} & 728 & 16.39 & 265 & 0.17 & 265 & 0.17 & -- & -- & -- & -- & \textbf{244} & 258.19 \\
8 & 8 & 43 &  & 235 & \textbf{0.04} & 171 & 4.44 & 103 & 0.05 & 96 & 0.07 & -- & -- & -- & -- & \textbf{64} & 22.38 \\
8 & 8 & 100 &  & 566 & \textbf{0.06} & 399 & 9.55 & 176 & 0.08 & 158 & 0.13 & -- & -- & -- & -- & \textbf{138} & 46.89 \\
8 & 8 & 150 &  & 895 & \textbf{0.07} & 577 & 13.23 & 233 & 0.11 & 224 & 0.16 & -- & -- & -- & -- & \textbf{191} & 93.68 \\
8 & 8 & 203 &  & 1153 & \textbf{0.09} & 721 & 16.44 & 263 & 0.12 & 262 & 0.21 & -- & -- & -- & -- & \textbf{227} & 92.61 \\
8 & 8 & 247 &  & 1412 & \textbf{0.09} & 835 & 19.37 & 263 & 0.12 & 263 & 0.17 & -- & -- & -- & -- & \textbf{256} & 113.55 \\
\midrule
10 & 2 & 8 &  & 20 & \textbf{0.02} & 20 & 0.26 & 30 & 0.03 & 26 & 0.04 & 16 & 600.00 & \textbf{10} & 21.00 & 15 & 3.70 \\
10 & 2 & 17 &  & 32 & \textbf{0.02} & 32 & 0.62 & 53 & 0.04 & 43 & 0.05 & -- & 600.00 & \textbf{13} & 600.00 & 27 & 11.71 \\
10 & 2 & 28 &  & 41 & \textbf{0.02} & 41 & 0.95 & 72 & 0.04 & 51 & 0.09 & -- & 600.00 & \textbf{21} & 600.00 & 43 & 600.00 \\
10 & 2 & 35 &  & 44 & \textbf{0.03} & 44 & 1.23 & 83 & 0.06 & 66 & 0.12 & -- & 600.00 & \textbf{33} & 600.00 & 53 & 39.88 \\
10 & 2 & 45 &  & \textbf{53} & \textbf{0.03} & \textbf{53} & 1.53 & 106 & 0.07 & 79 & 0.14 & -- & 600.00 & 65 & 600.00 & 70 & 40.20 \\ 
\bottomrule\end{tabular}
}\end{center}
\begin{center}
\centering
\adjustbox{max width=1.\textwidth}{\renewcommand{\arraystretch}{1.}
\begin{tabular}{ccccccccccccccccccccccccc}
\toprule
& & & & \multicolumn{4}{c}{Baselines} & \multicolumn{4}{c}{Ansatze} & \multicolumn{6}{c}{Binary linear programs} \\
\cmidrule(lr){5-8} \cmidrule(lr){9-12} \cmidrule(lr){13-18}
\multicolumn{3}{c}{Instances} & \hspace{0.2cm} & \multicolumn{2}{c}{Qiskit} & \multicolumn{2}{c}{Tket} & \multicolumn{2}{c}{\VT} & \multicolumn{2}{c}{Q-\VT} & \multicolumn{2}{c}{\CDP} & \multicolumn{2}{c}{\TB} & \multicolumn{2}{c}{\TBFM} \\
\cmidrule(lr){1-3} \cmidrule(lr){5-6} \cmidrule(lr){7-8} \cmidrule(lr){9-10} \cmidrule(lr){11-12} \cmidrule(lr){13-14} \cmidrule(lr){15-16} \cmidrule(lr){17-18}
$n$ & $d$ & $|M|$ &  & Depth & CPU [s] & Depth & CPU [s] & Depth & CPU [s] & Depth & CPU [s] & Depth & CPU [s] & Depth & CPU [s] & Depth & CPU [s] \\ \midrule
10 & 4 & 46 &  & 188 & \textbf{0.03} & 143 & 3.72 & 129 & 0.07 & 112 & 0.14 & -- & 600.00 & 110 & 600.00 & \textbf{69} & 227.05 \\
10 & 4 & 117 &  & 567 & \textbf{0.06} & 354 & 9.22 & 291 & 0.16 & 263 & 0.23 & -- & 600.00 & -- & 600.00 & \textbf{175} & 295.89 \\
10 & 4 & 157 &  & 799 & \textbf{0.06} & 460 & 12.05 & 355 & 0.25 & 337 & 0.32 & -- & 600.00 & -- & 600.00 & \textbf{220} & 600.00 \\
10 & 4 & 213 &  & 1066 & \textbf{0.08} & 549 & 14.50 & 434 & 0.25 & 429 & 0.42 & -- & 600.00 & -- & 600.00 & \textbf{293} & 446.20 \\
10 & 4 & 265 &  & 1343 & \textbf{0.10} & 645 & 16.61 & 492 & 0.31 & 492 & 0.46 & -- & 600.00 & -- & 600.00 & \textbf{332} & 592.32 \\
10 & 5 & 83 &  & 498 & \textbf{0.04} & 343 & 7.98 & 226 & 0.12 & 207 & 0.19 & -- & -- & -- & -- & \textbf{127} & 600.00 \\
10 & 5 & 161 &  & 952 & \textbf{0.07} & 585 & 13.46 & 394 & 0.20 & 349 & 0.36 & -- & -- & -- & -- & \textbf{245} & 600.00 \\
10 & 5 & 270 &  & 1618 & \textbf{0.10} & 888 & 20.95 & 550 & 0.30 & 511 & 0.41 & -- & -- & -- & -- & \textbf{374} & 600.00 \\
10 & 5 & 348 &  & 2087 & \textbf{0.13} & 1096 & 25.49 & 653 & 0.36 & 623 & 0.52 & -- & -- & -- & -- & \textbf{487} & 600.00 \\
10 & 5 & 427 &  & 2566 & \textbf{0.17} & 1300 & 29.25 & 721 & 0.41 & 699 & 0.65 & -- & -- & -- & -- & \textbf{590} & 600.00 \\
10 & 6 & 133 &  & 853 & \textbf{0.06} & 577 & 14.85 & 319 & 0.21 & 275 & 0.23 & -- & -- & -- & -- & \textbf{187} & 600.00 \\
10 & 6 & 233 &  & 1506 & \textbf{0.09} & 889 & 22.99 & 502 & 0.26 & 465 & 0.41 & -- & -- & -- & -- & \textbf{328} & 600.00 \\
10 & 6 & 372 &  & 2425 & \textbf{0.15} & 1342 & 34.16 & 669 & 0.41 & 649 & 0.61 & -- & -- & -- & -- & \textbf{498} & 600.00 \\
10 & 6 & 468 &  & 3071 & \textbf{0.20} & 1612 & 116.77 & 780 & 0.52 & 759 & 0.70 & -- & -- & -- & -- & \textbf{580} & 600.00 \\
10 & 6 & 595 &  & 3911 & \textbf{0.27} & 1915 & 59.23 & 872 & 0.52 & 872 & 0.77 & -- & -- & -- & -- & \textbf{682} & 600.00 \\
10 & 7 & 151 &  & 1038 & \textbf{0.06} & 679 & 18.28 & 349 & 0.21 & 313 & 0.25 & -- & -- & -- & -- & \textbf{219} & 600.00 \\
10 & 7 & 302 &  & 2054 & \textbf{0.12} & 1232 & 30.63 & 594 & 0.33 & 571 & 0.58 & -- & -- & -- & -- & \textbf{418} & 600.00 \\
10 & 7 & 473 &  & 3239 & \textbf{0.20} & 1748 & 44.42 & 798 & 0.46 & 765 & 0.59 & -- & -- & -- & -- & \textbf{610} & 600.00 \\
10 & 7 & 588 &  & 4113 & \textbf{0.27} & 2105 & 50.92 & 919 & 0.52 & 900 & 0.78 & -- & -- & -- & -- & \textbf{725} & 600.00 \\
10 & 7 & 757 &  & 5249 & \textbf{0.38} & 2557 & 62.11 & 987 & 0.59 & 985 & 0.81 & -- & -- & -- & -- & \textbf{838} & 600.00 \\
10 & 8 & 168 &  & 1296 & \textbf{0.08} & 790 & 20.73 & 409 & 0.22 & 369 & 0.28 & -- & -- & -- & -- & \textbf{243} & 600.00 \\
10 & 8 & 382 &  & 2750 & \textbf{0.17} & 1548 & 44.23 & 712 & 0.36 & 679 & 0.59 & -- & -- & -- & -- & \textbf{527} & 600.00 \\
10 & 8 & 524 &  & 3797 & \textbf{0.32} & 1996 & 53.69 & 846 & 0.64 & 828 & 0.65 & -- & -- & -- & -- & \textbf{668} & 600.00 \\
10 & 8 & 707 &  & 5118 & \textbf{0.49} & 2488 & 65.27 & 992 & 0.63 & 975 & 0.99 & -- & -- & -- & -- & \textbf{837} & 600.00 \\
10 & 8 & 892 &  & 6530 & \textbf{0.76} & 3050 & 79.09 & 1039 & 0.77 & 1039 & 1.14 & -- & -- & -- & -- & \textbf{952} & 600.00 \\
10 & 9 & 209 &  & 1573 & \textbf{0.12} & 991 & 29.59 & 475 & 0.26 & 434 & 0.41 & -- & -- & -- & -- & \textbf{311} & 600.00 \\
10 & 9 & 394 &  & 2916 & \textbf{0.20} & 1642 & 45.38 & 700 & 0.40 & 657 & 0.62 & -- & -- & -- & -- & \textbf{502} & 600.00 \\
10 & 9 & 570 &  & 4369 & \textbf{0.30} & 2281 & 58.59 & 903 & 0.51 & 870 & 0.75 & -- & -- & -- & -- & \textbf{721} & 600.00 \\
10 & 9 & 771 &  & 5856 & \textbf{0.55} & 2977 & 75.63 & 1035 & 0.58 & 1029 & 0.87 & -- & -- & -- & -- & \textbf{885} & 600.00 \\
10 & 9 & 977 &  & 7480 & 0.82 & 3537 & 86.11 & 1044 & \textbf{0.59} & 1044 & 1.00 & -- & -- & -- & -- & \textbf{1008} & 600.00 \\
10 & 10 & 187 &  & 1383 & \textbf{0.10} & 871 & 24.82 & 439 & 0.26 & 395 & 0.34 & -- & -- & -- & -- & \textbf{267} & 600.00 \\
10 & 10 & 387 &  & 2974 & \textbf{0.21} & 1675 & 44.43 & 696 & 0.41 & 668 & 0.47 & -- & -- & -- & -- & \textbf{522} & 600.00 \\
10 & 10 & 589 &  & 4615 & \textbf{0.39} & 2379 & 62.01 & 893 & 0.48 & 862 & 0.68 & -- & -- & -- & -- & \textbf{735} & 600.00 \\
10 & 10 & 807 &  & 6185 & \textbf{0.58} & 3054 & 79.46 & 1035 & 0.59 & 1029 & 0.91 & -- & -- & -- & -- & \textbf{911} & 600.00 \\
10 & 10 & 1013 &  & 7931 & 0.85 & 3723 & 91.77 & 1033 & \textbf{0.58} & 1033 & 0.99 & -- & -- & -- & -- & \textbf{1024} & 600.00 \\
\midrule
13 & 2 & 9 &  & 17 & \textbf{0.02} & 17 & 0.26 & 34 & 0.09 & 24 & 0.15 & -- & 600.00 & \textbf{10} & 106.59 & 16 & 14.55 \\
13 & 2 & 32 &  & 47 & \textbf{0.03} & 47 & 1.06 & 94 & 0.12 & 59 & 0.17 & -- & 600.00 & \textbf{25} & 600.00 & 55 & 54.28 \\
13 & 2 & 46 &  & 62 & \textbf{0.03} & 62 & 1.56 & 114 & 0.19 & 92 & 0.16 & -- & 600.00 & \textbf{56} & 600.00 & 69 & 600.00 \\
13 & 2 & 61 &  & \textbf{65} & \textbf{0.03} & \textbf{65} & 2.15 & 144 & 0.16 & 102 & 0.28 & -- & 600.00 & 75 & 600.00 & 91 & 600.00 \\
13 & 2 & 78 &  & \textbf{71} & \textbf{0.04} & \textbf{71} & 3.48 & 178 & 0.18 & 133 & 0.22 & -- & 600.00 & 79 & 600.00 & 124 & 176.19 \\
13 & 4 & 166 &  & 847 & \textbf{0.07} & 528 & 14.64 & 498 & 0.41 & 448 & 0.48 & -- & 600.00 & -- & 600.00 & \textbf{299} & 600.00 \\
13 & 4 & 312 &  & 1634 & \textbf{0.12} & 915 & 24.09 & 854 & 0.64 & 785 & 0.79 & -- & 600.00 & -- & 600.00 & \textbf{512} & 600.00 \\
13 & 4 & 480 &  & 2659 & \textbf{0.21} & 1252 & 37.16 & 1198 & 0.86 & 1145 & 1.08 & -- & 600.00 & -- & 600.00 & \textbf{809} & 600.00 \\
13 & 4 & 636 &  & 3540 & \textbf{0.27} & 1506 & 44.25 & 1434 & 1.06 & \textbf{1409} & 1.51 & -- & 600.00 & -- & 600.00 & -- & 600.00 \\
13 & 4 & 806 &  & 4537 & \textbf{0.33} & 1769 & 49.08 & \textbf{1686} & 1.27 & \textbf{1686} & 1.75 & -- & 600.00 & -- & 600.00 & -- & 600.00 \\
13 & 5 & 337 &  & 2148 & \textbf{0.18} & 1291 & 34.85 & 942 & 0.70 & 853 & 0.87 & -- & -- & -- & -- & \textbf{599} & 600.00 \\
13 & 5 & 665 &  & 4634 & \textbf{0.37} & 2277 & 57.29 & 1665 & 1.24 & \textbf{1546} & 1.94 & -- & -- & -- & -- & -- & 600.00 \\
13 & 5 & 980 &  & 6788 & \textbf{0.54} & 3092 & 76.62 & 2288 & 1.71 & \textbf{2155} & 2.76 & -- & -- & -- & -- & -- & -- \\
13 & 5 & 1336 &  & 9385 & \textbf{0.82} & 3989 & 99.73 & 2855 & 2.04 & \textbf{2743} & 3.42 & -- & -- & -- & -- & -- & -- \\
13 & 5 & 1664 &  & 11802 & \textbf{1.12} & 4689 & 121.57 & 3385 & 2.52 & \textbf{3203} & 4.09 & -- & -- & -- & -- & -- & -- \\
13 & 6 & 556 &  & 4404 & \textbf{0.25} & 2459 & 73.43 & 1527 & 1.19 & \textbf{1377} & 1.33 & -- & -- & -- & -- & -- & -- \\
13 & 6 & 1160 &  & 9332 & \textbf{0.71} & 4385 & 120.31 & 2694 & 2.58 & \textbf{2521} & 2.94 & -- & -- & -- & -- & -- & -- \\
13 & 6 & 1690 &  & 13728 & \textbf{1.28} & 5917 & 158.31 & 3539 & 2.92 & \textbf{3351} & 4.46 & -- & -- & -- & -- & -- & -- \\
13 & 6 & 2228 &  & 18178 & \textbf{1.96} & 7371 & 201.98 & 4247 & 4.40 & \textbf{4151} & 6.04 & -- & -- & -- & -- & -- & -- \\
13 & 6 & 2808 &  & 23141 & \textbf{3.16} & 8715 & 252.52 & \textbf{4773} & 4.46 & \textbf{4773} & 8.19 & -- & -- & -- & -- & -- & -- \\
13 & 7 & 819 &  & 7121 & \textbf{0.47} & 3823 & 110.50 & 2137 & 1.71 & \textbf{1879} & 2.21 & -- & -- & -- & -- & -- & -- \\
13 & 7 & 1643 &  & 14560 & \textbf{1.33} & 6798 & 194.67 & 3609 & 2.92 & \textbf{3301} & 4.54 & -- & -- & -- & -- & -- & -- \\
13 & 7 & 2506 &  & 22613 & \textbf{2.63} & 9874 & 285.00 & 4833 & 4.18 & \textbf{4562} & 8.85 & -- & -- & -- & -- & -- & -- \\
13 & 7 & 3251 &  & 29287 & \textbf{4.25} & 12180 & 380.19 & 5625 & 4.43 & \textbf{5438} & 10.42 & -- & -- & -- & -- & -- & -- \\
13 & 7 & 4095 &  & 37038 & 6.59 & 14580 & 418.87 & 6384 & \textbf{5.71} & \textbf{6258} & 14.13 & -- & -- & -- & -- & -- & -- \\
13 & 8 & 1072 &  & 9838 & \textbf{0.81} & 4835 & 144.58 & 2628 & 1.92 & \textbf{2359} & 3.09 & -- & -- & -- & -- & -- & -- \\
13 & 8 & 2151 &  & 19880 & \textbf{3.03} & 8785 & 239.21 & 4347 & 3.22 & \textbf{4029} & 5.76 & -- & -- & -- & -- & -- & -- \\
13 & 8 & 3227 &  & 30539 & 5.39 & 12339 & 362.78 & 5770 & \textbf{4.21} & \textbf{5520} & 9.97 & -- & -- & -- & -- & -- & -- \\
13 & 8 & 4360 &  & 41707 & 9.37 & 15605 & 486.82 & 6772 & \textbf{5.21} & \textbf{6652} & 13.80 & -- & -- & -- & -- & -- & -- \\
13 & 8 & 5382 &  & 51675 & 13.83 & 18288 & 523.48 & \textbf{7305} & \textbf{6.41} & \textbf{7305} & 17.53 & -- & -- & -- & -- & -- & -- \\
13 & 9 & 1298 &  & 12613 & \textbf{1.16} & 6372 & 185.78 & 3114 & 2.25 & \textbf{2780} & 3.38 & -- & -- & -- & -- & -- & -- \\
13 & 9 & 2579 &  & 25354 & \textbf{3.59} & 11409 & 338.76 & 4965 & 3.74 & \textbf{4626} & 7.47 & -- & -- & -- & -- & -- & -- \\
13 & 9 & 3925 &  & 39100 & 8.56 & 16087 & 555.62 & 6503 & \textbf{5.25} & \textbf{6255} & 11.49 & -- & -- & -- & -- & -- & -- \\
13 & 9 & 5205 &  & 52050 & 14.10 & 20235 & 600.00 & 7532 & \textbf{5.77} & \textbf{7381} & 15.44 & -- & -- & -- & -- & -- & -- \\
13 & 9 & 6526 &  & 65624 & 27.11 & 24201 & 600.00 & 7992 & \textbf{6.63} & \textbf{7976} & 18.89 & -- & -- & -- & -- & -- & -- \\
13 & 10 & 1537 &  & 15310 & \textbf{1.82} & 7586 & 227.87 & 3502 & 2.48 & \textbf{3146} & 4.54 & -- & -- & -- & -- & -- & -- \\
13 & 10 & 2943 &  & 30050 & 5.12 & 13219 & 409.48 & 5478 & \textbf{3.99} & \textbf{5085} & 8.28 & -- & -- & -- & -- & -- & -- \\
13 & 10 & 4465 &  & 46300 & 10.53 & 18668 & 600.00 & 7060 & \textbf{5.30} & \textbf{6835} & 13.55 & -- & -- & -- & -- & -- & -- \\
13 & 10 & 5964 &  & 62006 & 24.85 & 23575 & 600.00 & 8080 & \textbf{6.64} & \textbf{7985} & 17.70 & -- & -- & -- & -- & -- & -- \\
13 & 10 & 7384 &  & 77213 & 38.71 & 28044 & 600.00 & \textbf{8253} & \textbf{7.16} & \textbf{8253} & 19.30 & -- & -- & -- & -- & -- & -- \\
\bottomrule\end{tabular}
}\end{center}

\subsection{Low Autocorrelation Binary Sequences (LABS) instances}

\bigskip
\begin{center}
\centering
\adjustbox{max width=1.\textwidth}{\renewcommand{\arraystretch}{1.}
\begin{tabular}{cccccccccccccccccc}
\toprule
& & & & \multicolumn{4}{c}{Baselines} & \multicolumn{4}{c}{Ansatze} & \multicolumn{6}{c}{Binary linear programs} \\
\cmidrule(lr){5-8} \cmidrule(lr){9-12} \cmidrule(lr){13-18}
\multicolumn{3}{c}{Instances} & \hspace{0.2cm} & \multicolumn{2}{c}{Qiskit} & \multicolumn{2}{c}{Tket} & \multicolumn{2}{c}{\VT} & \multicolumn{2}{c}{Q-\VT} & \multicolumn{2}{c}{\CDP} & \multicolumn{2}{c}{\TB} & \multicolumn{2}{c}{\TBFM} \\
\cmidrule(lr){1-3} \cmidrule(lr){5-6} \cmidrule(lr){7-8} \cmidrule(lr){9-10} \cmidrule(lr){11-12} \cmidrule(lr){13-14} \cmidrule(lr){15-16} \cmidrule(lr){17-18}
$n$ & $d$ & $|M|$ &  & Depth & CPU [s] & Depth & CPU [s] & Depth & CPU [s] & Depth & CPU [s] & Depth & CPU [s] & Depth & CPU [s] & Depth & CPU [s] \\
\midrule
5 & 4 & 5 &  & 25 & 0.07 & 21 & 0.41 & 16 & \textbf{0.04} & 15 & 0.07 & \textbf{9} & 29.53 & 10 & 0.89 & 13 & 0.42 \\
5 & 4 & 7 &  & 35 & 0.02 & 29 & 0.39 & 20 & \textbf{0.01} & 20 & 0.02 & \textbf{11} & 600.00 & \textbf{11} & 0.90 & 14 & 0.44 \\
\midrule
6 & 4 & 7 &  & 35 & 0.02 & 29 & 0.44 & 23 & \textbf{0.01} & 22 & 0.02 & \textbf{10} & 600.00 & 11 & 4.91 & 15 & 1.05 \\
6 & 4 & 11 &  & 50 & 0.02 & 42 & 0.71 & 29 & \textbf{0.01} & 29 & 0.02 & -- & 600.00 & \textbf{14} & 18.16 & 19 & 1.37 \\
6 & 4 & 13 &  & 65 & 0.02 & 52 & 0.96 & 33 & \textbf{0.01} & 31 & 0.02 & 29 & 600.00 & \textbf{15} & 7.43 & 21 & 1.78 \\
\midrule
7 & 4 & 9 &  & 45 & 0.02 & 37 & 0.59 & 29 & \textbf{0.01} & 28 & 0.02 & -- & 600.00 & \textbf{12} & 33.40 & 18 & 1.90 \\
7 & 4 & 15 &  & 65 & 0.02 & 55 & 0.97 & 41 & \textbf{0.02} & 37 & 0.03 & -- & 600.00 & \textbf{16} & 600.00 & 26 & 3.00 \\
7 & 4 & 19 &  & 94 & \textbf{0.02} & 81 & 1.52 & 55 & 0.02 & 48 & 0.04 & -- & 600.00 & \textbf{18} & 600.00 & 32 & 7.77 \\
7 & 4 & 22 &  & 109 & \textbf{0.02} & 89 & 1.60 & 61 & 0.02 & 53 & 0.03 & -- & 600.00 & \textbf{21} & 187.70 & 32 & 8.19 \\
\midrule
8 & 4 & 11 &  & 55 & 0.02 & 45 & 0.70 & 34 & \textbf{0.02} & 32 & 0.03 & -- & 600.00 & \textbf{12} & 88.33 & 22 & 6.14 \\
8 & 4 & 19 &  & 80 & \textbf{0.02} & 68 & 1.26 & 51 & 0.03 & 48 & 0.04 & -- & 600.00 & \textbf{18} & 600.00 & 31 & 5.51 \\
8 & 4 & 25 &  & 121 & \textbf{0.02} & 104 & 2.02 & 71 & 0.03 & 63 & 0.04 & -- & 600.00 & \textbf{26} & 600.00 & 42 & 142.24 \\
8 & 4 & 31 &  & 143 & \textbf{0.02} & 117 & 2.23 & 81 & 0.04 & 76 & 0.05 & -- & 600.00 & 64 & 600.00 & \textbf{49} & 115.41 \\
8 & 4 & 34 &  & 170 & \textbf{0.02} & 131 & 2.60 & 83 & 0.04 & 78 & 0.05 & -- & 600.00 & 72 & 600.00 & \textbf{48} & 14.60 \\
\midrule
10 & 4 & 15 &  & 75 & \textbf{0.02} & 61 & 0.95 & 44 & 0.03 & 41 & 0.06 & -- & 600.00 & \textbf{14} & 600.00 & 28 & 9.42 \\
10 & 4 & 27 &  & 110 & \textbf{0.02} & 94 & 1.83 & 75 & 0.05 & 66 & 0.07 & -- & 600.00 & 52 & 600.00 & \textbf{45} & 17.44 \\
10 & 4 & 37 &  & 167 & \textbf{0.02} & 148 & 2.90 & 117 & 0.10 & 96 & 0.14 & -- & 600.00 & 100 & 600.00 & \textbf{61} & 48.74 \\
10 & 4 & 49 &  & 203 & \textbf{0.03} & 171 & 3.60 & 140 & 0.07 & 128 & 0.11 & -- & 600.00 & 129 & 600.00 & \textbf{79} & 148.14 \\
10 & 4 & 58 &  & 271 & \textbf{0.03} & 226 & 4.60 & 162 & 0.09 & 148 & 0.14 & -- & 600.00 & 153 & 600.00 & \textbf{92} & 211.14 \\
10 & 4 & 66 &  & 300 & \textbf{0.03} & 239 & 5.33 & 180 & 0.09 & 166 & 0.12 & -- & 600.00 & 165 & 600.00 & \textbf{100} & 600.00 \\
10 & 4 & 70 &  & 326 & \textbf{0.03} & 248 & 5.85 & 185 & 0.10 & 173 & 0.19 & -- & 600.00 & 178 & 600.00 & \textbf{99} & 600.00 \\
\midrule
13 & 4 & 21 &  & 105 & \textbf{0.02} & 85 & 1.38 & 59 & 0.10 & 56 & 0.15 & -- & 600.00 & 49 & 600.00 & \textbf{36} & 40.82 \\
13 & 4 & 39 &  & 155 & \textbf{0.02} & 133 & 2.64 & 105 & 0.12 & 96 & 0.17 & -- & 600.00 & 78 & 600.00 & \textbf{60} & 61.19 \\
13 & 4 & 55 &  & 230 & \textbf{0.03} & 205 & 4.44 & 175 & 0.23 & 161 & 0.23 & -- & 600.00 & 129 & 600.00 & \textbf{91} & 235.33 \\
13 & 4 & 76 &  & 293 & \textbf{0.03} & 249 & 5.63 & 232 & 0.19 & 208 & 0.26 & -- & 600.00 & 189 & 600.00 & \textbf{123} & 542.26 \\
13 & 4 & 94 &  & 425 & \textbf{0.04} & 360 & 7.38 & 292 & 0.26 & 259 & 0.34 & -- & 600.00 & -- & 600.00 & \textbf{159} & 600.00 \\
13 & 4 & 114 &  & 475 & \textbf{0.04} & 390 & 8.50 & 340 & 0.27 & 317 & 0.34 & -- & 600.00 & -- & 600.00 & \textbf{193} & 600.00 \\
13 & 4 & 130 &  & 538 & \textbf{0.04} & 440 & 10.88 & 414 & 0.34 & 380 & 0.36 & -- & 600.00 & -- & 600.00 & \textbf{245} & 600.00 \\
\bottomrule
\end{tabular}
}\end{center}

\end{document}